\documentclass[10pt,a4paper,twocolumn]{ieeetran} 
\usepackage{cite}
\usepackage{enumerate}
\usepackage{stmaryrd}
\usepackage{comment}
\usepackage{hyperref}
\usepackage{amsmath,amssymb,amsthm}

\def\R{{\mathbb R}}
\def\N{{\mathbb N}}

\def\C{\mathcal{C}}
\def\F{\mathcal{F}}
\def\H{\mathcal{H}}
\def\L{\mathcal{L}}

\def\I{\mathcal{I}}

\def\md{\boldsymbol{\tau}} 

\def\zero{\mathbf{0}}
\def\ess{\mathrm{ess.}}

\newcommand\Lint[1]{\L(#1)}
\newcommand\Nin[1]{\mathcal{N}\big(#1\big)}
\def\wto{\rightharpoonup}
\def\tforall{\text{for all }}
\def\f{\mathbf{f}}

\def\xx{\mathsf{x}}
\def\zz{\mathsf{z}}

\newcommand{\mex}{$\hfill\circ$} 

\newcommand{\conc}[1]{\substack{\displaystyle\frown\\ {\,\scriptstyle #1\,}}}
\newcommand{\norm}[1]{\llbracket #1\rrbracket}

\DeclareMathOperator*{\esssup}{\mathrm{ess.sup}}

\newtheorem{thm}{Theorem}[section]
\newtheorem{defin}[thm]{Definition}
\newtheorem{prop}[thm]{Proposition}
\newtheorem{cor}[thm]{Corollary}
\newtheorem{lem}[thm]{Lemma}

\newtheorem{claim}{Claim}
\newtheorem{as}[thm]{Assumption}
\newtheorem{rem}[thm]{Remark}

\begin{document}
\title{Time-Delay Systems with Discrete and Distributed delays:\\
  Discontinuous Initial Conditions and Reachability Sets}
\author{Hernan Haimovich, \IEEEmembership{Senior Member, IEEE}, Jos\'e L. Mancilla-Aguilar
\thanks{Supported in part by Agencia I+D+i PICT 2021-I-A-0730, Argentina.}
\thanks{H. Haimovich is with the International Center for Systems and Information Science (CIFASIS), CONICET--UNR, Argentina (e-mail: haimovich@cifasis-conicet.gov.ar).}
\thanks{J.L. Mancilla-Aguilar is with the Facultad de Ingenier\'{\i}a, Universidad de Buenos Aires, Argentina (e-mail: jmancill@fi.uba.ar).}
} 

\maketitle

\begin{abstract}
Time-invariant finite-dimensional systems, under reasonable continuity assumptions, exhibit the property that if solutions exist for all future times, the set of vectors reachable from a bounded set of initial conditions over bounded time intervals is also bounded. This property can be summarized as follows: forward completeness implies bounded reachability sets. By contrast, this property does not necessarily hold for infinite-dimensional systems in general, and time-delay systems in particular. Sufficient conditions for this property to hold that can be directly tested on the function defining the system dynamics are only known in the case of systems with pointwise (or discrete) delays. This paper develops novel sufficient conditions for the boundedness of the reachability sets of time-delay systems involving mixed pointwise and distributed delays. Broad classes of systems satisfying these conditions are identified. 
\end{abstract}


\section{Introduction}
\label{sec:introduction} 

\subsection{Time-delay Systems with Mixed Delays}

In the context of systems theory, the concept of state for a time-delay system is given by the history of the evolution of the ``state'' vector (which is not the true system state) over a time period whose length equals the maximum delay involved. The standard theory of time-delay systems considers initial conditions that are continuous functions \cite{halver_book93,karjia_book11,fridman_book14}. Properties of time-delay systems that involve the behavior of sets of solutions, as opposed to single solutions, may require the consideration, albeit in a limiting sense, of solutions arising from discontinuous initial conditions \cite{manhai_auto24}. One such property is the boundedness of reachability sets—namely, that for any bounded set of initial conditions, the corresponding set of reachable vectors remains bounded over any compact time interval.

The boundedness of reachability sets (BRS), very related to robust forward completeness \cite{karjia_book11}, is an essential property that is required or implied by many type of stability results \cite{karpep_jco22,chakar_mcss23}, and is fundamentally necessary for input-to-state stability (ISS) related properties \cite{mirwir_cdc17,mirwir_tac18,mirpri_siamrev20}. For time-invariant finite-dimensional systems satisfying specific (Lipschitz) continuity assumptions, the mere fact that all maximally defined solutions exist for all future times, a property called forward completeness (FC), is actually equivalent to BRS \cite{linson_jco96}. For infinite-dimensional systems, including those involving delays, BRS is known to be stronger than FC \cite{mirwir_tac18,manhai_auto24}. Infinite-dimensional systems for which FC holds but not BRS may lack standard bounds on the state norm given by well-known stability properties such as ISS, see e.g. \cite{schmid_mcss19}. 


Time-delay systems containing both discrete and distributed delays are considered in this paper, defined by a retarded functional differential equation of the general form
\begin{align}
    \label{eq:f-intro}
  \dot x(t) = \f(t,x_t),
\end{align}
with $x_t$ the evolution history at time $t$ of the state vector $x$, which is a function with domain $[-\md,0]$ and whose value is given by $x_t(s) = x(t+s) \in \R^n$, and with $\md$ the maximum delay. This type of time-delay system is said to involve only constant discrete delays when $\f$ can be written as
\begin{align*}
  \f(t,x_t) &= g(t,x_t(0),x_t(-\tau_1),\ldots,x_t(-\tau_\ell))\\
  &= g(t,x(t),x(t-\tau_1),\ldots,x(t-\tau_\ell))
\end{align*}
for some function $g$. If this is not possible, then the system can be said to involve nonconstant and/or distributed delays.

In~\cite{manhai_auto24}, it was shown that time-invariant time-delay systems involving only constant discrete delays and having Lipschitz continuous dynamics exist that exhibit FC but not BRS. The fact that FC is not equivalent to BRS enables the lack of other connections known to hold for time-invariant systems without delay. For example, a time-delay system may be globally asymptotically stable but not uniformly so \cite{manhai_auto24} or not even uniformly globally attractive \cite{chawir_lcss24}. In~\cite{bricha_lcss24}, it was shown that FC and BRS become equivalent if FC holds for initial conditions in a set larger than that comprised of only continuous functions. \cite{manhai_auto24}, \cite{chawir_lcss24} and~\cite{bricha_lcss24} consider time-delay systems having only constant discrete delays. Consideration of larger sets of initial conditions brings about important technical details related to existence and uniqueness of solutions. In this regard, the results of \cite{bricha_lcss24} apply some results of \cite{delmit_jde72} where the Cauchy problem for time-delay systems is formulated allowing for discontinuous initial conditions in Banach spaces related to $L^p$. However, most of the results in \cite{delmit_jde72} are not valid for $p=\infty$. In particular, note that the formulation of the Cauchy problem in equations (3.8) to (3.10) therein is said to make sense by Proposition 4.3, which does not hold for $p=\infty$.

In this paper, the BRS of systems with discrete and distributed delays is analyzed, with initial conditions in a specific Banach space that allows for discontinuous initial conditions (the precise formulation will become clear in Section~\ref{sec:equiv-init-cond}). The motivation for this consideration arises from the fact that BRS, as opposed to FC, requires the analysis of sets and sequences of solutions rather than single solutions, and that bounded sequences of continuous initial conditions may converge to discontinuous functions. In \cite{mironc_scl23}, the property called robust forward completeness, which is stronger than BRS but equivalent to it under globally bounded inputs, is characterized by means of necessary and sufficient conditions, including some based on Lyapunov-type functions. The aim of the current paper is to, alternatively, provide results that can be checked based on knowledge of the function defining the system dynamics, i.e. the function $\f$ in \eqref{eq:f-intro}, and suited especially to time-varying systems with mixed pointwise and distributed delays. 



Examples of nonlinear systems with distributed or mixed delays that cannot exceed some maximum value 
can be found in
\cite{cookap_mb76, willeg_siamjma82, guolak_jmaa88, kertes_amh90, bendap_book07, Chrif2015, aleefi_scl24}.
\cite{cookap_mb76} considers nonlinear systems with distributed delays under a periodicity condition.
\cite{willeg_siamjma82} and \cite{guolak_jmaa88} analyze infectious disease models that are nonlinear and include finite distributed delays.
\cite{kertes_amh90} addresses a type of nonlinear system with mixed delays, but where the distributed delays form part of the coefficients of a linear system with pointwise delays.
The book \cite{bendap_book07} contains many examples of delay systems of very general forms but where linearity plays a central role. 
The Nicholson’s blowflies model considered in \cite[eq.~(1.2)]{Chrif2015} contains several concentrated as well as distributed delays and is nonlinear.
A nonlinear distributed delay system with homogeneity properties is considered in \cite{aleefi_scl24}. The publications \cite{cookap_mb76, willeg_siamjma82, guolak_jmaa88, kertes_amh90, bendap_book07, Chrif2015, aleefi_scl24} consider only continuous initial conditions.

In this context, our main contributions are the following. First, we provide checkable conditions under which the mere existence of solutions is sufficient to ensure the BRS of time-varying systems with mixed discrete and distributed delays. Second, we identify broad classes of systems that satisfy the conditions, and relate these to examples in the literature. To the best of our knowledge, the results provided are completely novel and constitute a basis for the subsequent analysis of BRS for mixed-delay systems with inputs.

The remainder of this paper is organized as follows. Section~\ref{sec:notation} introduces the notation employed throughout. Section~\ref{sec:syst-exist-uniq} covers essential existence and uniqueness results for time-delay systems with initial conditions not restricted to continuous functions. Sufficient conditions for the boundedness of the reachability sets of nonlinear time-varying systems with mixed delays are given in Section~\ref{sec:FCandBRS}. 
Conclusions are drawn in Section~\ref{sec:conclusions}. The appendix contains some important technical proofs.

\subsection{Notation}
\label{sec:notation}

The reals, nonnegative reals, naturals and naturals including zero are denoted by $\R$, $\R_{\ge0}$, $\N$ and $\N_0$. For $a,b\in\R$, $a \vee b := \max\{a,b\}$. For $x\in\R^n$, $|x|$ denotes the infinity norm of $x$, {\emph{i.e.}} $|x|=\max_{i=1,\ldots,n}|x_i|$, $x_i$ being the $i$-th component of $x$. For every $R\ge 0$ and $n\in\N$, $B_R^n := \{x\in\R^n : |x|\le R\}$ denotes the closed ball of radius $R$ in $\R^n$. 
For a matrix $A\in \R^{p\times q}$, $A^T$ denotes its transpose and $\|A\|=\max_{|x|=1}|Ax|$ its induced operator norm. 
Given $\md\ge 0$, $\L^\infty := \L^\infty([-\md,0],\R^n)$ denotes the set of essentially bounded Lebesgue measurable functions $\phi : [-\md,0] \to \R^n$, $L^\infty:=L^\infty([-\md,0],\R^n)$ denotes the set of equivalence classes of functions in $\L^\infty([-\md,0],\R^n)$ where two functions are in the same equivalence class if they coincide almost everywhere, and $\C:=\C\left([-\md,0],\R^{n}\right)$ denotes the set of continuous functions $\phi : [-\md,0] \to \R^n$, so that $\C \subset \L^{\infty}$. The supremum and essential supremum of a function $\phi \in \L^\infty$ are distinguished by including the subscript~`$_\infty$' for the pointwise supremum: $\|\phi\|_{\infty} := \sup_{-\md\le t \le 0} |\phi(t)|$ and $\|\phi\| := \ess\sup_{-\md\le t \le 0} |\phi(t)|$. Note then that $\|\phi\|\le \|\phi\|_{\infty} \le \infty$ and $\|\phi\| < \infty$ for every $\phi\in\L^\infty$. For $\psi \in L^\infty$, $\|\psi\|$ denotes $\|\phi\|$ with $\phi\in\L^\infty$ any representative of the equivalence class $\psi$.
Also, $L^1 := L^1([-\md,0],\R^n)$ is the set of equivalence classes of Lebesgue measurable functions $g$ that satisfy $\|g\|_1 := \sum_{i=1}^n\int_{-\md}^0 |g_i(s)| ds < \infty$, where two functions are in the same equivalence class if they differ on a set of measure zero.
Closed balls of radius $R$ in $L^1$, $L^\infty$ and $\L^\infty$ are denoted by $B_R^{L^1} :=  \{ g \in L^1 : \| g\|_1 \le R \}$, $B_R^{L^\infty} := \{ \phi \in L^\infty : \|\phi\| \le R \}$ and $B_R^{\L^\infty} := \{ \phi \in \L^\infty : \|\phi\|_{\infty} \le R \}$. 
For $a\in \R$, let $\I_a$ be the class of intervals $J$ of the form $J=[a,b)$ or $[a,b]$ with $a<b\le \infty$.
For a function $\xx:J\to \R^n$, 
with $J\in \I_{t_0-\md}$, and any $t\in [t_0,\infty)\cap J$, 
$\xx_t$ is defined as the function $\xx_t:[-\md,0]\to \R^n$ satisfying $\xx_{t}(s)=\xx(t+s)$ for all $s\in [-\md,0]$. Given $\phi \in \mathcal{L}^{\infty}$ and $x:J\to \R^n$, with $J\in \I_{t_0}$, $(\phi\conc{t_0} x) : [t_0-\md,t_0)\cup J \to \R^n$ denotes the concatenation of $\phi$ and $x$, defined as 
\begin{align}
    \label{eq:conc-def}
    (\phi\conc{t_0} x)(t) &= 
    \begin{cases}
        \phi(t-t_0) &\text{if }t\in [t_0-\md,t_0),\\
        x(t) &\text{if }t\in J.
    \end{cases}
\end{align}
Given a function $\xx:J\to \R^n$ and numbers $\tau_i \in \R_{\ge 0}$ for $i=0,\ldots,\ell$, the notation $\xx(t-\tau_{i:j})$ is used as abbreviation for $(\xx(t-\tau_i),\xx(t-\tau_{i+1}),\ldots,\xx(t-\tau_j)) \in \R^{(j-i+1)n}$.
A function $g:I\times \R^q \to \R^k$, with $I\subset \R$, is said to be Carath\'eodory if it is Lebesgue measurable in the first argument and continuous in the others.

\section{System, Existence and Uniqueness}
\label{sec:syst-exist-uniq}

This section introduces the precise type of system considered and covers essential existence and uniqueness results. In Section~\ref{sec:essential-time-delay}, delay systems in a very general form of retarded functional differential equation are described, with initial conditions in $\L^\infty$. In Section~\ref{sec:existence-solutions}, corresponding existence and uniqueness results are provided. Section~\ref{sec:discrete-and-distributed} particularizes the system model to one allowing for a combination of discrete and distributed delays. Section~\ref{sec:equiv-init-cond} proves that the assumptions enable the consideration of a specific Banach space as a suitable set of initial conditions. 

\subsection{Time-delay Systems}
\label{sec:essential-time-delay}

Consider time-delay systems of the form
\begin{equation}
  \begin{split}
    \label{eq:td-gen}
    \dot{x}(t) &= \f(t, (\phi\conc{t_0} x)_t),
    \qquad t\ge t_0,\\
    \qquad x(t_0) &= \phi(0),
    \qquad \phi\in \L^\infty,
  \end{split}
\end{equation}
with $\f : \R^n \times \L^\infty \to \R^n$.
In this paper, a solution of~(\ref{eq:td-gen}) is a locally absolutely continuous function $x : J \to \R^n$, with $J\in \I_{t_0}$, which satisfies
\begin{align}
    \label{eq:td-gen-sol}
    x(t) = \phi(0) + \int_{t_0}^t \f(s, (\phi\conc{t_0} x)_s ) ds
\end{align}
for all $t\in J$. The time $t_0$ in~(\ref{eq:td-gen}) is called the initial time and the function $\phi$ is the initial condition at time $t_0$. A solution from this initial condition is one that satisfies~(\ref{eq:td-gen-sol}). A solution $x:J \to \R^n$ corresponding to the initial time $t_0$ and the initial condition $\phi$ is maximally defined (in forward time) if there is no other solution $y : J^{\prime} \to \R^n$ with the same initial time and initial condition that is a strict extension of $x$. 
  
\subsection{Existence and Uniqueness of Solutions} 
\label{sec:existence-solutions}

\begin{as}
  \label{as:Lip} The function $\f:\R_{\ge 0} \times \L^\infty \to \R^n$ satisfies the following:
  \begin{enumerate}[i)] 
 \item For every $R\ge 0$ and $T>0$, there exist $L=L(R,T) \ge 0$ and a zero-measure set $Z \subset \R_{\ge 0}$ such that
  \begin{align*}
    |\f(t,\phi^1) - \f(t,\phi^2)| \le L \|\phi^1 - \phi^2\|_{\infty} 
  \end{align*}
  for all $\phi^1,\phi^2 \in B_{R}^{\L^{\infty}}$ and $t\in [0,T] \setminus Z$;\label{item:fphicont}
  \item $\f(\cdot,\zero)$ is locally essentially bounded;\label{item:f0bound}
  \item for every $t_0\ge 0$, $T>t_0$, $\phi \in \L^\infty$ and continuous $x:[t_0,T]\to \R^n$, the function $t\mapsto \f(t,(\phi\conc{t_0} x)_t)$ is Lebesgue measurable.\label{item:ftmeas}
  \end{enumerate}
\end{as}
The constant $L=L(R,T)$ of Assumption~\ref{as:Lip}\ref{item:fphicont}) is referred to as the Lipschitz constant of $\f$ corresponding to $R$ and $T$. 
\begin{rem}
    When considering initial conditions in 
    $\L^p$, the set of Lebesgue measurable functions $\phi$ for which $\sum_{i=1}^n \int_{-\md}^0 |\phi_i(s)|^p ds<\infty$, with $1 \le p < \infty$, and provided item~\ref{item:fphicont}) holds with $\|\cdot\|_{\infty}$ replaced by $\|\cdot\|_p$, then item~\ref{item:ftmeas}) can be replaced by $\f(\cdot, \phi)$ being measurable for every $\phi$. In that case, the current item~\ref{item:ftmeas}) would follow taking into account that the map $t \mapsto (\phi\conc{t_0}x)_t$, with values in $\L^p$, is continuous \cite[Proposition~4.3(i)]{delmit_jde72}. However, the measurability requirement \ref{item:ftmeas}) is unavoidable when considering initial conditions in $\L^\infty$ (or $L^\infty$, see Section~\ref{sec:equiv-init-cond}).
\end{rem}
\begin{prop}
  \label{prop:exist-uniq}
  Let $\f$ satisfy Assumption~\ref{as:Lip}. Then, for each $t_0\ge 0$ and $\phi\in \L^\infty$ with $\|\phi\|_{\infty} < \infty$, there is a unique maximally defined solution $x : [t_0,T_x) \to \R^n$ of (\ref{eq:td-gen}). Moreover, if $\sup_{t_0 \le t < T_x} |x(t)| < \infty$, then $T_x = \infty$.
\end{prop}
\begin{proof} 
  Let $t_0\ge 0$, $\phi\in \L^{\infty}$ and $\alpha \in (t_0 , t_0+1]$.
  Consider the operator $P : \C_{\alpha} \to \C_{\alpha}$, where $\C_{\alpha}$ is the set of continuous functions $y:[t_0,\alpha]\to \R^n$, defined by
  \begin{align}
    \label{eq:Pop-def}
  Py(t) &= 
    \phi(0) + \int_{t_0}^t \f\left(s, (\phi\conc{t_0} y)_s\right) ds 
  \end{align}
  Note that $P$ is well-defined due to Assumption~\ref{as:Lip}\ref{item:ftmeas}).
  For any $y \in \C_{\alpha}$ and $R\ge \max\{\|y\|_{\infty},\|\phi\|_{\infty}\}$, it is true that
\begin{align*}
  \esssup_{s\in [t_0,\alpha]}\left| \f\big(s, (\phi\conc{t_0} y)_s\big) \right| &\le \esssup_{t_0\le t \le t_0+1}|\f(t,\zero)| + L R =: F_R
\end{align*}
where Assumption~\ref{as:Lip}\ref{item:fphicont}) has been used, with $L$ the Lipschitz constant of $\f$ corresponding to $R$ and $T=t_0+1$, and where $F_R$ is finite due to Assumption~\ref{as:Lip}\ref{item:f0bound}).
The operator $P$ satisfies
\begin{align*}
  |Py(t)| &\le |\phi(0)| + \int_{t_0}^t \left| \f\big(s, (\phi\conc{t_0} y)_s\big) \right| ds\\
  &\le \|\phi\|_{\infty} + F_R (t-t_0)\quad\text{for all }t\in [t_0,\alpha].
\end{align*}
Let $R > \|\phi\|_{\infty}$ and assume that, in addition, $\alpha \le t_0+(R-\|\phi\|_{\infty})/F_R$. Then, for every $y$ satisfying $\|y\|_{\infty}\le R$, it follows that $\|Py\|_{\infty}\le R$ and thus $P : U \to U$ with $U$ defined as $U = \{y\in \C_{\alpha} : \|y\|_{\infty} \le R\}$. We claim that the restriction of $P$ to $U$ is a contraction if $\alpha$ also satisfies the condition $\lambda:=L(\alpha-t_0)<1$.
Let $y,z \in U$. For every $t_0\le t \le \alpha$, then
\begin{align*}
  |Py(t)-Pz(t)| \le \int_{t_0}^t \big|\f\big(s, (\phi\conc{t_0} y)_s\big)-\f\big(s, (\phi\conc{t_0} z)_s\big)\big| ds \\
  \le L(\alpha-t_0) \|(\phi\conc{t_0} y)_s-(\phi\conc{t_0} z)_s\|_{\infty} 
  \le \lambda \|y-z\|_{\infty},
\end{align*}
therefore
\begin{align*}
\|Py-Pz\|_{\infty}\le \lambda \|y-z\|_{\infty}
\end{align*}
and hence the claim follows.
Since $U$ is closed and $\C_{\alpha}$ is a Banach space, $U$ is a complete metric space; therefore, Banach's fixed-point theorem establishes that $P$ has a unique fixed point $\tilde x \in U$. This fixed point must satisfy
\begin{align*}
  \tilde x(t) = \phi(0) + \int_{t_0}^t \f(s, (\phi\conc{t_0} \tilde x)_s ) ds,\quad t\in [t_0,\alpha],
\end{align*}
is hence absolutely continuous on $[t_0,\alpha]$ and therefore a solution of (\ref{eq:td-gen}). 

It is now proven that for any $t_0\ge 0$ and bounded $\phi\in \L^{\infty}$, there exists at least a solution $x$ of (\ref{eq:td-gen}) which is defined in some interval $[t_0,\alpha]$ with $\alpha>t_0$. Standard arguments of the theory of functional differential equations ensure that every solution of (\ref{eq:td-gen}) can be extended to a maximal solution $x$ whose domain of definition is of the form $[t_0,T_x)$ with $T_x$ possibly infinite.

Next, take two maximally defined solutions $x: [t_0,T_x) \to \R^n$ and $y : [t_0,T_y) \to \R^n$ from the same bounded initial condition $\phi \in \L^{\infty}$ and let $T = \min\{T_x,T_y\}$ and $T' \in (t_0,T)$. Let $M=\max_{t\in [t_0,T']}\max\{|x(t)|,|y(t)|\}$
and let $L$ be the Lipschitz constant of $\f$ corresponding to $M$ and $T'$.
The difference $z=x-y$ satisfies, for $t_0\le r \le t \le T'$
\begin{align*}
  &|z(r)| \le \int_{t_0}^{r} \left| \f\Big(s, (\phi\conc{t_0} x)_s\Big) - \f\Big(s, (\phi\conc{t_0} y)_s\Big) \right| ds\\
  &\le \int_{t_0}^{t} L \|(\phi\conc{t_0} x)_s-(\phi\conc{t_0} y)_s\|_{\infty}\ ds 
  \le \int_{t_0}^{t} L \max_{t_0\le \tau\le s}|z(\tau)| ds
\end{align*}
where Assumption~\ref{as:Lip}\ref{item:fphicont}) has been used, taking into account that the value of the integral is not affected by the value of the integrand on the zero-measure set $Z$.
Therefore, the function $g(t):=\max_{t_0\le r\le t}|z(r)|$ satisfies
\begin{align}
  \label{eq:g-delta-int}
  g(t)\le \int_{t_0}^t L\: g(s) \:ds\quad \text{for all}\;t \in [t_0,T']
\end{align}
and Gronwall's lemma implies that $g(t)=0$ for all $t\in [t_0,T']$. Then, $x(t)=y(t)$ for $t\in [t_0,T']$. Since $T'<T$ is arbitrary, $x(t)=y(t)$ for all $t\in [t_0,T)$. If $T_x \neq T_y$, then the latter would contradict the fact that one of the solutions is maximally defined. As a consequence, $T_x = T_y$ and any maximally defined solution is unique.

  Finally, let $\sup_{t_0 \le t < T_x} |x(t)| < \infty$ and suppose for a contradiction that $T_x < \infty$. From Assumption~\ref{as:Lip}, then $x$ is Lipschitz and hence uniformly continuous on $[t_0,T_x)$, and $\lim_{t\to T_x^-} x(t)$ exists. This means that $x$ can be extended to a solution on $[t_0,T_x]$, which contradicts $x$ being maximally defined. Therefore, $T_x = \infty$.
\end{proof} 
\begin{rem}
The proof of Proposition~\ref{prop:exist-uniq} follows the standard arguments used to establish existence and uniqueness for Lipschitz retarded functional differential equations; see, for example, \cite[Chapter~4]{bendap_book07}.   
\end{rem}

\subsection{Combining Discrete and Distributed Delays}
\label{sec:discrete-and-distributed}

Time-delay systems with a combination of discrete and distributed delays will be considered, as follows. 
Let $\f$ in~(\ref{eq:td-gen}) be of the following form
\begin{align}
  \label{eq:f-pt-dist-p}
  \f(t,\phi) = \tilde f(t,(\phi(0),\phi(-\tau_1),\ldots,\phi(-\tau_\ell)),\phi)
\end{align}
with $0 < \tau_1 < \cdots < \tau_\ell = \md$ and $\tilde f : \R_{\ge 0} \times \R^{(\ell+1) n} \times \L^\infty \to \R^n$ satisfying
\begin{align}
    \label{eq:zerom-same}
  \| \phi - \psi \| = 0 \quad\Rightarrow\quad 
  \tilde f(t,\xi_{0:\ell},\phi) = \tilde f(t,\xi_{0:\ell},\psi) 
\end{align}
for all $t\in\R_{\ge 0}$, all
\begin{align}
  \label{eq:xi-vector}
  \xi_{0:\ell} := (\xi_0,\xi_1,\ldots,\xi_\ell)\in\R^{(\ell+1) n},
\end{align}
and all $\phi,\psi \in \L^\infty$. The function $\tilde f$ models the specific way in which pointwise or discrete delays on the one hand, and distributed delays on the other, affect the system. According to~\eqref{eq:zerom-same}, the effect of the distributed delay, via the last argument of $\tilde f$, is insensitive to differences over zero-measure sets. Therefore, there exists $f : \R_{\ge 0} \times \R^{(\ell+1) n} \times L^\infty \to \R^n$ such that
\begin{align}
  \label{eq:f-pt-dist}
  f(t,\xi_{0:\ell},\bar\phi) = \tilde f(t,\xi_{0:\ell},\phi)
\end{align}
with $\bar\phi \in L^\infty$ the equivalence class of $\phi \in \L^\infty$. Hereafter, no distinction will be made between an element $\phi\in \L^\infty$ and its equivalence class in $L^\infty$ whenever no ambiguity arises. Therefore, the system
\begin{subequations}
  \label{eq:system-pt-dist}
  \begin{align}
    &\dot x(t) = f(t,\xx(t-\tau_{0:\ell}),\xx_t)\\
    \label{eq:xxtau-def}
    &\xx(t-\tau_{0:\ell}) = (\xx(t-\tau_0),\ldots,\xx(t-\tau_\ell)),\quad
    \tau_0 = 0,\\
    \label{eq:xx-def}
    &\xx = (\phi\conc{t_0}x),\quad
    x(t_0) = \phi(0),\quad
    \phi \in \L^\infty
  \end{align}
\end{subequations}
defined for $t\ge t_0$, is well formulated, even if $\phi\in\L^\infty$ but the last argument of $f$ should be in $L^\infty$.
The following assumption ensures that~(\ref{eq:system-pt-dist}) has well-defined and unique maximally defined solutions.
\begin{as}
    \label{as:Lip-pt}
    The function $f : \R_{\ge 0} \times \R^{(\ell+1)n} \times L^\infty \to \R^n$ satisfies:
    \begin{enumerate}[i)]
    \item For every $R> 0$ and $T>0$, there exist $L=L(R,T)\ge 0$ and zero-measure set $Z \subset \R_{\ge 0}$ such that\label{item:fLip-pt-dist}
        \begin{align*}
            |f(t,\xi_{0:\ell},\phi) - f(t,\zeta_{0:\ell},\psi)| \le L \left(| \xi_{0:\ell} - \zeta_{0:\ell} | \vee \|\phi - \psi\|\right)
        \end{align*}
        for all $\xi_{0:\ell},\zeta_{0:\ell} \in B_R^{(\ell+1)n}$, $\phi,\psi \in B_R^{L^\infty}$, $t\in [0,T]\setminus Z$; 
    \item $f(\cdot,0,\zero)$ is locally essentially bounded;\label{item:f0bnd-pt-dist}
    \item for every $t_0 \ge 0$, $T>t_0$, $\phi\in\L^\infty$ and continuous $x : [t_0,T] \to \R^n$, the function
        $$t \mapsto f(t,(\xx(t-\tau_0),\xx(t-\tau_1),\ldots,\xx(t-\tau_\ell)),\xx_t)$$
        is Lebesgue measurable, where $\xx = (\phi\conc{t_0}x)$.\label{item:ftt-meas}
    \end{enumerate}
\end{as}
A broad class of nonlinear systems that satisfy Assumption~\ref{as:Lip-pt} under simpler checkable conditions is the following:
\begin{subequations}
    \label{eq:sys-exa1}
    \begin{align}        
        \dot x(t) &= g(t,\xx(t-\tau_{0:\ell}),\Nin{t,\xx_t}),\\
        \label{eq:functional-def}
        \Nin{t,\psi} &= \int_{-\md}^0 G(t,s,\psi(s)) ds,
    \end{align}    
\end{subequations}
    with $\xx(t-\tau_{0:\ell})$ and $\xx$ as in \eqref{eq:xxtau-def}--\eqref{eq:xx-def}, $G:\R_{\ge 0}\times [-\md,0]\times \R^n\to \R^p$,
    $\mathcal{N}: \R_{\ge 0} \times L^\infty \to \R^p$, $g:\R_{\ge 0}\times \R^{(\ell+1)n}\times \R^p\to \R^n$.
    %
\begin{prop}
    \label{prop:exa:ass2} 
    The function $f$ defined via $f(t,\xi_{0:\ell},\phi) = g(t,\xi_{0:\ell},\Nin{t,\phi})$ with $g$ as above satisfies Assumption~\ref{as:Lip-pt} whenever the following hold:
    \begin{enumerate}[a)]
    \item \label{ex:caratheodory} $g$ is Carath\'eodory and $g(\cdot,0,0)$ is locally essentially bounded;
    \item \label{ex:lipschitz} for every $R> 0$ and $T>0$, there exist $\bar L=\bar L(R,T)\ge 0$ and a zero-measure set $Z\subset \R_{\ge 0}$ so that
        \begin{align*}
            |g(t,\xi_{0:\ell},\nu) - g(t,\bar\xi_{0:\ell},\bar \nu)|\le \bar L \left(| \xi_{0:\ell} - \bar\xi_{0:\ell} |\vee |\nu - \bar \nu|\right)
        \end{align*}
        for all $\xi_{0:\ell},\bar \xi_{0:\ell} \in B_R^{(\ell+1)n}$, $\nu,\bar \nu \in B_R^{p}$, $t\in [0,T]\setminus Z$; 
    \item \label{item:caratheodory-G} $G(\cdot,s,\eta)$ is continuous for every $s$ and $\eta$, and for each fixed $t$, the mapping $G(t,\cdot,\cdot)$ is Carath\'eodory;
    \item \label{item:Lips-G} for every $R\ge 0$ and $T>0$, there exist $\hat L_{R,T}$ and $M_{R,T}\in L^1([-\md,0],\R)$ such that
        \begin{align*}
            |G(t,s,\eta)|\le M_{R,T}(s) \; &\text{and}\\
            |G(t,s,\eta) - G(t,s,\hat \eta)| &\le \hat L_{R,T}(s)  |\eta - \hat \eta| 
        \end{align*}
        for all $t\in [0,T]$, $\eta, \hat \eta \in B_R^n$, and almost all $s\in [-\md,0]$. 
    \end{enumerate}
\end{prop} 
The most intricate part of the proof, given in Appendix~\ref{app:ex:ass2}, is to ensure the measurability requirement of item~\ref{item:ftt-meas}) of Assumption~\ref{as:Lip-pt}.


\subsection{Equivalent Initial Conditions}
\label{sec:equiv-init-cond}

From an applications' standpoint, it is reasonable to request that mathematical models representing real systems involving time delays should have the following property: initial conditions $\phi \in \L^\infty$ that differ on sets of measure zero [with the exception of the value $\phi(0)$, recall~(\ref{eq:td-gen}) and~(\ref{eq:td-gen-sol})] should give rise to the same solution. Different initial conditions that generate the same solution can conceptually be regarded as ``equivalent''. Sufficient conditions for the latter to hold are next provided. The following lemma is straightforward.
\begin{lem}
  \label{lem:equiv-ass}
  Consider $f : \R_{\ge 0} \times \R^{(\ell+1)n} \times L^\infty \to \R^n$ and $\f$ of the form~(\ref{eq:f-pt-dist-p})--(\ref{eq:f-pt-dist}). Then, Assumption~\ref{as:Lip-pt} implies Assumption~\ref{as:Lip}.
\end{lem}
For functions $f$ that satisfy Assumption~\ref{as:Lip-pt}, the solutions of system~(\ref{eq:system-pt-dist}) exist and are unique, according to Lemma~\ref{lem:equiv-ass} and Proposition~\ref{prop:exist-uniq}.
The following result shows that under Assumption~\ref{as:Lip-pt} the solutions of~(\ref{eq:system-pt-dist}) are in addition insensitive to initial conditions $\phi\in\L^\infty$ that differ on sets of measure zero, provided that the value $\phi(0)$ is the same.
\begin{prop}
  \label{prop:insensitive}
  Let $f : \R_{\ge 0} \times \R^{(\ell+1)n} \times L^\infty \to \R^n$ satisfy Assumption~\ref{as:Lip-pt}. Consider the system~(\ref{eq:system-pt-dist}) with $t_0 \ge 0$ and two bounded initial conditions $\phi^1,\phi^2\in\L^\infty$. Let $x^1 : [t_0,T_1) \to \R^n$, $x^2 : [t_0,T_2) \to \R^n$ be the corresponding maximally defined solutions. If $\|\phi^1 - \phi^2\| = 0$ and $\phi^1(0)=\phi^2(0)$, then $x^1 \equiv x^2$.
\end{prop}
\begin{proof}
  Let $z = x^1 - x^2$, $\xx^1 = (\phi^1\conc{t_0}x^1)$, $\xx^2 = (\phi^2\conc{t_0}x^2)$. Let $T = \min\{T_1,T_2\}$, $T' \in (t_0,T)$, 
  \begin{align*}
    R = \max_{i=1,2}\left\{\max\left\{\|\phi^i\|_\infty, \max_{t_0 \le t \le T'} |x^i(t)|\right\} \right\}
  \end{align*}
  and $L$ be in correspondence with $R$ and $T$, according to Assumption~\ref{as:Lip-pt}\ref{item:fLip-pt-dist}).
  Then, for $t\in [t_0,T']$,
  \begin{align*}
    &|z(t)| \le\\
    &\int_{t_0}^t \left| f(s,\xx^1(s-\tau_{0:\ell}),\xx^1_s) - f(s,\xx^2(s-\tau_{0:\ell}),\xx^2_s)\right| ds \\
    &\le \int_{t_0}^t L \left(\max_{i=0,\ldots,\ell} |\xx^1(s-\tau_i) - \xx^2(s-\tau_i)| \vee \| \xx^1_s - \xx^2_s \|\right) ds
  \end{align*}
  Since $\|\phi^1 - \phi^2 \| = 0$, we have, for all $s\in [t_0,t]$,
  \begin{align*}
    &\|\xx^1_s - \xx^2_s\| = \max_{t_0 \le r \le s} |z(r)|\\
    \intertext{and, for all $i=0,\ldots,\ell$ and almost all $s \in [t_0,t]$,}
    |\xx^1(s-\tau_i) &- \xx^2(s-\tau_i)| \le \max_{t_0 \le r \le s} |z(r)| 
  \end{align*}
  Defining $g(t) = \max_{t_0 \le r \le t} |z(r)|$, it follows that
  \begin{align}
  \label{eq:g-delta-int-2}
  g(t)\le \int_{t_0}^t L\: g(s) \:ds\quad \text{for all}\;t \in [t_0,T'].
\end{align}
  From Gronwall inequality it follows that $0 = g(t) \ge |z(t)| = |x^1(t)-x^2(t)|$ for all $t \in [t_0,T']$. Since $T'$ can be made arbitrarily close to $T$, this means that $x^1(t) = x^2(t)$ for all $t\in [t_0,T)$. Then, $T_1 = T_2$ or otherwise either $x^1$ or $x^2$ would not be maximally defined.
\end{proof}


Proposition~\ref{prop:exist-uniq} ensures existence and uniqueness of solutions only for bounded initial conditions $\phi\in\L^\infty$, meaning that the pointwise supremum $\|\phi\|_\infty$ should be finite. For a system of the form~(\ref{eq:system-pt-dist}) where $f$ satisfies Assumption~\ref{as:Lip-pt}, Proposition~\ref{prop:insensitive} establishes the equivalence, in terms of the generated solutions, of initial conditions that differ on sets of measure zero while coinciding at $0$. Therefore, if $\phi\in\L^\infty$ is unbounded, one may define the solution generated by such initial condition as the solution generated by any bounded $\psi\in\L^\infty$ that differs from $\phi$ on a set of measure zero and satisfies $\psi(0)=\phi(0)$. This $\psi$ always exists because the elements of $\L^\infty$ are, by definition, essentially bounded. In the sequel, the following system will hence be considered
\begin{subequations}
  \label{eq:system-pt-dist-ess}
  \begin{align}
    &\dot x(t) = f(t,\xx(t-\tau_{0:\ell}),\xx_t),\quad \xx = (\phi\conc{t_0}x), \\
    &\xx(t-\tau_{0:\ell}) = (\xx(t-\tau_0),\ldots,\xx(t-\tau_\ell)),\quad
    \tau_0 = 0,\\
    &x(t_0) = \xi_0 \in \R^n,\quad
    \phi \in L^\infty,
  \end{align}
\end{subequations}
defined for $t\ge t_0 \ge 0$, where $f$ satisfies Assumption~\ref{as:Lip-pt}. A maximal solution $x:[t_0,T_x)\to \R^ n$ of (\ref{eq:system-pt-dist-ess}) is a maximal solution of (\ref{eq:system-pt-dist}) corresponding to any initial condition $\psi\in \L^{\infty}$ that is a representative of $\phi$ and satisfies $\psi(0)=\xi_0$. The difference with respect to~(\ref{eq:system-pt-dist}) is the fact that the initial condition $\phi\in\L^\infty$ of~(\ref{eq:system-pt-dist}) is replaced by the pair $(\xi_0,\phi)\in \R^n \times L^\infty$. We say that the initial condition $(\xi_0,\phi)$ is continuous if there exists $\psi\in \C$ such that $\psi$ is a representative of $\phi$ and $\psi(0)=\xi_0$. In such a case, the function $\psi \in \C$ is unique, we identify $(\xi_0,\phi)$ with such $\psi$ and write $(\xi_0,\phi)\in \C$. According to the preceding results, under Assumption~\ref{as:Lip-pt} the system~(\ref{eq:system-pt-dist-ess}), which combines discrete and distributed delays, is well-defined and its solutions exist and are unique. Let $x(\cdot;t_0,\xi_0,\phi)$ and $[t_0,T_{(t_0,\xi_0,\phi)})$ denote the unique solution of (\ref{eq:system-pt-dist-ess}) and its maximal interval of definition, respectively. The following property of the solutions can be straightforwardly proved.
\begin{prop} \label{prop:semigroup}
Let $f$ in (\ref{eq:system-pt-dist-ess}) satisfy Assumption \ref{as:Lip-pt}. Let  $t_0\ge 0$, $\xi_0\in \R^n$ and $\phi \in L^{\infty}$, and let $t_1\in [t_0,T_{(t_0,\xi_0,\phi)})$. Then, if $x(\cdot)=x(\cdot;t_0,\xi_0,\phi)$, $x_1=x(t_1)$ and $\phi_1=(\phi\conc{t_0}x)_{t_1}$, $T_{(t_1,\xi_1,\phi_1)}=T_{(t_0,\xi_0,\phi)}$ and $x(t;t_1,\xi_1,\phi_1)=x(t)$ for all $t\in [t_1,T_{(t_1,\xi_0,\phi)})$
\end{prop}
\begin{rem}
    The system models considered in \cite{cookap_mb76,willeg_siamjma82,guolak_jmaa88} are all of the form~\eqref{eq:sys-exa1}. By the above discussion, the corresponding solutions are well-defined for not only continuous initial conditions but also more general ones in $\R^n \times L^\infty$. By contrast, the model given by \cite[eq.~(2.1)]{kertes_amh90} is of the form~\eqref{eq:td-gen} and, under the assumptions given therein, cannot be ensured to depend on a mix of discrete and distributed delays in the sense considered here, even though the appearance of the equations would suggest so. The model in \cite[eq.~(1)]{aleefi_scl24} also is of the form~\eqref{eq:sys-exa1} but does not necessarily satisfy the Lipschitz continuity requirement of Proposition~\ref{prop:exa:ass2}~\ref{item:Lips-G}).
    \mex
\end{rem}

\section{Forward Completeness and Bounded Reachability Sets}
\label{sec:FCandBRS}

The main aim of this section is to establish sufficient conditions that guarantee the equivalence between the existence of solutions over compact time intervals and the boundedness of the reachability sets, \emph{ i.e.} the uniform  boundedness of the solutions, from bounded sets of initial conditions and over such time intervals. Toward this aim, the concept of weak-* convergence in $L^\infty$ will be fundamental.
\begin{defin}
    \label{def:brs}
System (\ref{eq:system-pt-dist-ess}), with $f$ satisfying Assumption~\ref{as:Lip-pt}, is said to have bounded reachability sets (BRS) if for every initial time $t_0 \ge 0$, the solutions of (\ref{eq:system-pt-dist-ess}) are defined for all $t\ge t_0$, \emph{i.e.} the system is forward complete, and
for all $0\le t_0<T$ and all $R>0$,
 \begin{align}
    \label{eq:brs}
    \sup_{t \in [t_0,T], \xi \in B_R^n, \phi \in B_R^{L^\infty}} |x(t;t_0,\xi,\phi)| < \infty.
  \end{align}
\end{defin}

\subsection{Weak-* Topology}
\label{sec:weak-star}

\begin{defin}
  \label{def:weakstar-conv}
  A sequence $\{\phi^k\}$ of elements of $L^\infty$ 
  is said to converge weakly-* to $\phi^0 \in L^\infty$ if for every $g\in L^1$, 
  \begin{align*}
    \lim_{k\to\infty} \int_{-\md}^0 g(s)^T \phi^k(s) ds = \int_{-\md}^0 g(s)^T \phi^0(s) ds
  \end{align*}
  The notation $\phi^k \wto \phi^0$ will be used as equivalent for $\{\phi^k\}$ converges weakly-* to $\phi^0$.
\end{defin}
The following lemma will be repeatedly used.
\begin{lem}
  \label{lem:weakfacts}
  The following facts are true:
  \begin{itemize}
  \item If $\phi^k \wto \phi^0$, then $\sup_{k\in\N} \|\phi^k\| < \infty$ and $\|\phi^0\| \le \liminf_{k\to\infty} \|\phi^k\|$.
  \item Every bounded sequence in $L^\infty$ contains a weakly-* convergent subsequence.
  \end{itemize}
\end{lem}
\begin{proof}
  The proof of the first item follows from item (iii) of Proposition 3.13 of~\cite{brezis_book11}. The second item is a consequence of the compactness of any closed ball in the weak-* topology, as established by the Banach-Alaoglu-Bourbaki theorem \cite[Theorem~3.16]{brezis_book11} (see also \cite[Corollary~3.30]{brezis_book11}). 
\end{proof}

Since $L^\infty$ 
is the dual space of $L^1$ 
\cite[Remark~5, p.~99]{brezis_book11}, and the latter is a separable space \cite[Theorem~4.13]{brezis_book11}, it follows that within any bounded ball of $L^\infty$, the weak-* topology is induced by a metric. Specifically, separability implies that there exists a countable set of functions $\{g^i : i\in\N\}\subset B_1^{L^1}$, 
that is dense in $B_1^{L^1}$. Let $R>0$ and consider the closed ball $B_R^{L^\infty}$. Following the proof of Theorem~3.28 of \cite{brezis_book11}, 
the following is a norm in $L^\infty$:
\begin{align}
  \label{eq:normstar}
  \norm{\phi} := \sum_{i=1}^\infty \frac{1}{2^i} \left| \int_{-\md}^0 g^i(s)^T \phi(s) ds \right|
\end{align}
For each $R>0$, the metric $d : B_R^{L^\infty} \times B_R^{L^\infty} \to \R_{\ge 0}$
\begin{align}
  \label{eq:def:metric}
  d(\phi,\psi) = \norm{\phi-\psi}
\end{align}
induces the weak-* topology on $B_R^{L^{\infty}}$. 
\begin{rem}
  If a sequence $\{\phi^k\}$ is bounded, 
  then $\phi^k \wto \phi^0 \in L^\infty$ if and only if
  \begin{align*}
    \lim_{k\to\infty} \norm{\phi^k - \phi^0} = \lim_{k\to\infty} d(\phi^k,\phi^0) = 0
  \end{align*}
\end{rem}

The next 
result establishes that when every function in a sequence $\{\phi^k\}$ is concatenated with the same bounded function, then weak-* convergence is uniform over time shifts.
\begin{lem}
  \label{lem:unifdist}
  Let $t_0 \ge 0$, $T>t_0$, $\phi^k \wto \phi^0 \in L^\infty$, $x:[t_0,T] \to \R^n$ measurable and essentially bounded, and $$R\ge \max\left\{ \esssup_{t\in [t_0,T]} |x(t)|, \sup_{k\in\N} \|\phi^k\| \right\}.$$
  Then, for every $\varepsilon > 0$ there exists $w\in\N$ such that
  \begin{align}
    \label{eq:unifdist}
    q \in [0,T-t_0],\ k\ge w \quad \Rightarrow \quad 
    \norm{\xx^k_{t_0+q} - \xx^0_{t_0+q}} < \varepsilon.
  \end{align}
  where $\xx^k = (\phi^k \conc{t_0} x)$.
\end{lem}
\begin{proof}
      If $q > \md$, then $\xx^k_{t_0+q} = \xx^0_{t_0+q} = x_{t_0+q}$, so it is sufficient to consider the case $q \in [0,\md]$. For each $i\in\N$ and $q \in [0,\md]$, define the functions $g^i_q \in L^1$ as
    \begin{align*}
      g^i_q(s) &= 
      \begin{cases}
        g^i(s-q) &\text{if }s \in [-\md+q,0],\\
        0 &\text{otherwise,}
      \end{cases}
    \end{align*}
    with $g^i$ as in~(\ref{eq:normstar}), so that $\|g^i_q\|_1 \le \|g^i\|_1 \le 1$ and
    \begin{align*}
      \lefteqn{\norm{\xx^k_{t_0+q} - \xx^0_{t_0+q}} = \norm{(\phi^k\conc{t_0} x)_{t_0+q} - (\phi^0\conc{t_0} x)_{t_0+q}}}\hspace{10mm}\\
      &= \sum_{i=1}^\infty \frac{1}{2^i} \left| \int_{-\md}^{-q} g^i(s)^T[\phi^k(s+q)-\phi^0(s+q)] ds \right|\\
      &= \sum_{i=1}^\infty \frac{1}{2^i} \left| \int_{-\md}^{0} g^i_q(s)^T [\phi^k(s)-\phi^0(s)] ds \right|.
    \end{align*}
    For a contradiction, suppose that there exists $\varepsilon > 0$ and sequences $\{k_r\}$, $\{q_{r}\}$ with $\{k_r\}$ increasing, $q_{r} \in [0,\min\{\md,T-t_0\}]$ such that $q_r \to q_0 \in [0,\min\{\md,T-t_0\}]$, and 
    \begin{align}
      \label{eq:eps2c}
      \norm{\xx^{k_r}_{t_0+q_{r}} - \xx^0_{t_0+q_{r}}} \ge \varepsilon.
    \end{align}
    Since $\|g^i_q\|_1 \le 1$ for all $q\in [0,\md]$ and $\|\phi^k-\phi^0\| \le 2R$ by the triangle inequality and Lemma~\ref{lem:weakfacts}, there exists $j\in\N$ such that for all $k\in\N$ and $q\in [0,\min\{\md,T-t_0\}]$,
    \begin{align*}
      \sum_{i=j+1}^\infty \frac{1}{2^i} \left| \int_{-\md}^{0} g^i_q(s)^T [\phi^k(s)-\phi^0(s)] ds \right| < \varepsilon/3
    \end{align*}
    Since $\phi^k \wto \phi^0$, there exists $p \in \N$ such that for all $k\ge p$,
    \begin{align*}
      \left| \int_{-\md}^{0} g^i_{q_0}(s)^T[\phi^k(s)-\phi^0(s)] ds \right| < \varepsilon/3
    \end{align*}
    for $i\in \{1,\ldots,j\}$.
    In addition, there exists $r_1 \in \N$ such that
    \begin{align*}
      \left| \int_{-\md}^{0} [g^i_{q_r}(s) - g^i_{q_0}(s)]^T [\phi^k(s)-\phi^0(s)] ds \right| < \varepsilon/3
    \end{align*}
    holds for all $r\ge r_1$, $i\in \{1,\ldots,j\}$, $k\in\N$.
    The latter is a consequence of $\lim_{r\to\infty} \|g^i_{q_r}(s) - g^i_{q_0}(s) \|_1 = 0$ for every $i\in\N$ and of $\|\phi^k-\phi^0\| \le 2R$. The following must then hold
    \begin{align*}
      \lefteqn{\norm{(\phi^k\conc{t_0} x)_{t_0+q} - (\phi^0\conc{t_0} x)_{t_0+q}} =}\hspace{5mm}\\
      &\sum_{i=1}^j \frac{1}{2^i} \left| \int_{-\md}^{0} [g^i_q(s)-g^i_{q_0}(s)+g^i_{q_0}(s)]^T [\phi^k(s)-\phi^0(s)] ds \right|\\
      &  + \sum_{i=j+1}^\infty \frac{1}{2^i} \left| \int_{-\md}^{0} g^i_q(s)^T[\phi^k(s)-\phi^0(s)] ds \right|
      < \varepsilon
    \end{align*}
    for some $q=q_r$, $k=k_r$, which contradicts~(\ref{eq:eps2c}).
\end{proof}

\subsection{Uniform Convergence of Solutions}

With the future aim of analyzing the behavior of solution sets for systems with inputs, families of time-varying systems without inputs are considered. A family of systems is defined not by a single function $f$ but by an \emph{arbitrary} set $\F$ of functions $f : \R_{\ge 0} \times \R^{(\ell+1)n} \times L^\infty \to \R^n$. 
\begin{defin}
  \label{def:family-Lip-pt}
  A set $\F$ of functions $f : \R_{\ge 0} \times \R^{(\ell+1)n} \times L^\infty \to \R^n$ is said to satisfy Assumption~\ref{as:Lip-pt} uniformly whenever the constant $L(R,T)$ in item~\ref{item:fLip-pt-dist}) and the local essential bound in item~\ref{item:f0bnd-pt-dist}) are the same for every $f \in \F$. The zero-measure set $Z$ may indeed depend on $f$.
\end{defin}
%
\begin{as}
  \label{as:fweaklim}
  The family $\F$ satisfies Assumption~\ref{as:Lip-pt} uniformly  and the following conditions. 
  \begin{enumerate}[a)]
  \item \label{item:weakstar-cont} for every $T>0$, $R>0$,
    $\varepsilon > 0$ and sequence $\{\phi^k\} \subset B_R^{L^\infty}$ satisfying $\phi^k \wto \phi^0$, there exists $j \in \N$ such that 
    \begin{align*}
      |f(t, \xi_{0:\ell}, \phi^k) - f(t, \xi_{0:\ell}, \phi^0)| < \varepsilon
    \end{align*}
    for all $t\in [0,T]\setminus Z_f$, $\xi_{0:\ell}\in B_R^{(\ell+1) n}$, $k\ge j$ and $f\in\F$, for some zero-measure set $Z_f\subset [0,T]$ that may depend on $f$;
  \item \label{item:discint-cont} for every $t_0 \ge 0$, 
  $R\ge 0$, continuous function $x : [t_0, t_0+\md] \to B_R^n$, sequence $\{\phi^k\} \subset B_R^{L^\infty}$ satisfying $\phi^k \wto \phi^0$ and $\varepsilon > 0$, there exists $j \in \N$ such that for all $t\in [t_0,t_0+\md]$, $k\ge j$, and $f\in\F$,
    \begin{align*}
      \left|\int_{t_0}^{t} \big[{\scriptstyle f(s,\xx^k(s-\tau_{0:\ell}),\xx^0_s) - f(s,\xx^0(s-\tau_{0:\ell}),\xx^0_s)}\big]\ ds \right| < \varepsilon
    \end{align*}
    with $\xx^k = (\phi^k\conc{t_0}x)$.
  \end{enumerate}
\end{as}
Assumption~\ref{as:fweaklim}\ref{item:weakstar-cont}) can be equivalently formulated by saying that the map $\phi \mapsto f(t,\xi_{0:\ell},\phi)$, which characterizes the distributed delay effect, is continuous in the weak-* topology uniformly over $\xi_{0:\ell}$ in compact sets and uniformly over $t$ almost everywhere in compact intervals. Assumption~\ref{as:fweaklim}\ref{item:discint-cont}) in turn imposes a condition on how the discrete delay may affect the dynamics. A broad class of systems satisfying this assumption is 
\begin{subequations}
    \label{eq:exa:ass3}
    \begin{align}
        \label{eq:exa:ass3:f}
        \dot x(t) &= g_0(t,x(t),\Lint{t,\xx_t}) + \notag\\
        &\hspace{1cm}G_1(t,x(t),\Lint{t,\xx_t})\xx(t-\tau_{1:\ell}) \\
        \label{eq:Lint:def}
        \Lint{t,\phi} &=\int_{-\md}^0 K(t,s)\phi(s)\:ds
    \end{align}    
\end{subequations}    
where $K:\R_{\ge 0}\times [-\md,0] \to \R^{p\times n}$, $\L : \R_{\ge 0}\times L^{\infty}\to \R^p$,
    $g_0:\R_{\ge 0}\times \R^n \times \R^p\to \R^n$, $G_1:\R_{\ge 0}\times \R^n \times \R^p\to \R^{n\times \ell n}$ and $\xx(t-\tau_{1:\ell})=(\xx(t-\tau_1),\ldots,\xx(t-\tau_{\ell}))$. 
    The following is our first contribution, whose relevance will become evident in Section~\ref{sec:BRS} with Theorem~\ref{thm:brs}.
\begin{prop}
    \label{prop:exa:ass3}
    The single-element family $\F=\{f\}$ with $f$ defined via 
    \begin{align*}
        f(t,\xi_{0:\ell},\phi)=g_0(t,\xi_0,\L(t,\phi))+G_1(t,\xi_0,\L(t,\phi)) \xi_{1:\ell},
    \end{align*}    
    where $\xi_{0:\ell}=(\xi_0,\xi_{1:\ell})$ and $\xi_{1:\ell}=(\xi_1,\ldots,\xi_{\ell})$,
    satisfies Assumption \ref{as:fweaklim}, whenever the following conditions hold, wherein $g_j$ denotes the $j$-th column of $G_1$: 
    \begin{enumerate}[a)] \label{ex:measK}
    \item $g_0,\ldots,g_{\ell n}$ satisfy conditions \ref{ex:caratheodory}) and \ref{ex:lipschitz}) of Proposition~\ref{prop:exa:ass2}, the latter with $\xi_0$ and $\bar{\xi}_0$ instead of $\xi_{0:\ell}$ and $\bar{\xi}_{0:\ell}$, respectively. 
%
    \item $K(\cdot,s)$ is continuous for all $s\in [-\md,0]$, and the columns of $K(t,\cdot)$ are measurable for all $t\ge 0$;
    \item \label{ex:boundK} for every $T>0$ there exists $M_T\in L^{1}([-\md,0],\R)$ such that
        $\|K(t,s)\|\le M_T(s)$
    for all $t\in [0,T]$ and almost all $s\in [-\md,0]$.
\end{enumerate}
\end{prop}
The main differences in the system~\eqref{eq:exa:ass3} compared to \eqref{eq:sys-exa1} are that $\Lint{t,\phi}$ is linear in $\phi$, as opposed to $\Nin{t,\phi}$, and that the differential equation is affine in the discrete delays. 

\begin{proof}
The function $f$ can be written in the form $f(t,\xi_{0:\ell},\phi)=g(t,\xi_{0:\ell},\Lint{t,\phi})$ with $g(t,\xi_{0:\ell},\nu)=g_0(t,\xi_0,\nu)+G_1(t,\xi_0,\nu)\xi_{1:\ell}$ and $\mathcal{L}(t,\phi)=\int_{-{\md}}^0 G(t,s,\phi(s)) ds$, with $G(t,s,\eta)=K(t,s)\eta$. The functions $g$ and $G$ satisfy the conditions in Proposition~\ref{prop:exa:ass2}. In consequence, $f$ satisfies Assumption \ref{as:Lip-pt}. 

Let $T>0$ and $R>0$. We first prove that $\mathcal{L}$ is continuous in $[0,T]\times B^{L^{\infty}}_R$ when $B^{L^{\infty}}_R$ is endowed with the weak* topology. Consider sequences $\{t^k\}\subset[0,T]$ and $\{\phi^k\}\subset B^{L^{\infty}}_R$ such that $t^k\to t^0$ and $\phi^k \wto \phi^0$. Then
\begin{multline*}
|\L(t^k,\phi^k)-\L(t^0,\phi^0)|\\
\le |\L(t^k,\phi^k)-\L(t^0,\phi^k)|+|\L(t^0,\phi^k)-\L(t^0,\phi^0)|.
\end{multline*}
Since
\begin{align*}
|\L(t^k,\phi^k)-\L(t^0,\phi^k)|&\le \int_{-\md}^0 \|K(t^k,s)-K(t^0,s)\|\|\phi^k\| ds\\
&\le R \int_{-\md}^0 \|K(t^k,s)-K(t^0,s)\|ds,
\end{align*}
$\|K(t^k,s)-K(t^0,s)\|\to 0$ for all $s$ due to the continuity of $K$ in $t$, and $\|K(t^k,s)-K(t^0,s)\|\le 2 M_T(s)$ for almost all $s$, by the Lebesgue dominated convergence theorem, it follows that 
$|\L(t^k,\phi^k)-\L(t^0,\phi^k)|\to 0$. Let $k_i(\cdot)$ be the i-th row of $K(t^0,\cdot)$. Then, $k_i\in L^1$ by item~\ref{ex:boundK}). Therefore,
$\int_{-\md}^0 k_i(s)[\phi^k(s)-\phi^0(s)]ds\to 0$ since $\phi^k\wto \phi^0$. In consequence $\int_{-\md}^0 K(t,s)[\phi^k(s)-\phi^0(s)]ds\to 0$ and hence $|\L(t^0,\phi^k)-\L(t^0,\phi^0)|\to 0$. So, $|\L(t^k,\phi^k)-\L(t^0,\phi^0)|\to 0$, and the continuity of $\L$ follows. Since $B^{L^{\infty}}_R$ is compact in the weak* topology (Lemma~\ref{lem:weakfacts}), $\L$ is bounded and uniformly continuous in $[0,T]\times B^{L^{\infty}}_R$. In particular, for every $\varepsilon>0$ there exists $\delta(\varepsilon)>0$ so that for all $\phi,\psi \in B^{L^{\infty}}_R$ and all $t\in [0,T]$
\begin{align*}
    \norm{\phi-\psi}<\delta \Rightarrow |\L(t,\phi)-\L(t,\psi)|<\varepsilon.
\end{align*}
We next prove that $\F=\{f\}$ satisfies item \ref{item:weakstar-cont}) of Assumption \ref{as:fweaklim}. Let $T>0$, $R>0$,
    $\varepsilon > 0$ and $\{\phi^k\} \subset B_R^{L^\infty}$ be a sequence satisfying $\phi^k \wto \phi^0$. Let 
    $$R^*=\max\big\{R,\ \sup\{|\L(t,\phi)|:t\in[0,T],\phi\in B_R^{L^\infty}\}\big\}$$
    and let $L=L(R^*,T)$ and $Z$ be the Lipschitz constant and the zero-measure set coming from item \ref{ex:lipschitz}) in Proposition~\ref{prop:exa:ass2} corresponding to $R^*$ instead of $R$. Then
\begin{multline*}
|f(t, \xi_{0:\ell}, \phi^k) - f(t, \xi_{0:\ell}, \phi^0)| \\
\le |g(t, \xi_{0:\ell}, \mathcal{L}(t,\phi^k)) - g(t, \xi_{0:\ell}, \mathcal{L}(t,\phi^0))|\\
\le L |\mathcal{L}(t,\phi^k)-\mathcal{L}(t,\phi^0)|
\end{multline*}
for all $\xi_{0:\ell}\in B^{(\ell+1)n}_R$ and all $t\in [0,T]\setminus Z$. By taking $\delta$ as above, but corresponding to $\varepsilon/L$ instead of $\varepsilon$, and $k$ so that $\norm{\phi^j-\phi^0}<\delta$ for all $j\ge k$, we have that
\begin{align*}
|f(t, \xi_{0:\ell}, \phi^k) - f(t, \xi_{0:\ell}, \phi^0)|<\varepsilon 
\end{align*}
for all $\xi_{0:\ell}\in B^{(\ell+1)n}_R$, all $t\in [0,T]\setminus Z$ and all $j\ge k$.

Next, we proceed to prove that $f$ satisfies item \ref{item:discint-cont}) of Assumption \ref{as:fweaklim}.
Let $t_0 \ge 0$, $R\ge 0$, $x : [t_0, t_0+\md] \to B_R^n$ continuous, $\{\phi^k\} \subset B_R^{L^\infty}$ a sequence satisfying $\phi^k \wto \phi^0$ and $\varepsilon > 0$. Define $\xx^k = (\phi^k\conc{t_0}x)$ for $k=0,1,\ldots.$ Then, for all $t\in [t_0,t_0+\md]$,
    \begin{multline*}
      \left|\int_{t_0}^{t} \big[{\scriptstyle f(s,\xx^k(s-\tau_{0:\ell}),\xx^0_s) - f(s,\xx^0(s-\tau_{0:\ell}),\xx^0_s)}\big]\ dt \right|\\
      =\left|\int_{t_0}^t {\scriptstyle G_1(s,x(s),\L(s,\xx^0_s))[\xx^k(s-\tau_{1:\ell})-\xx^0(s-\tau_{1:\ell})] } ds \right|\\
      = \left|\int_{t_0}^t {\scriptstyle H(s)[\xx^k(s-\tau_{1:\ell})-\xx^0(s-\tau_{1:\ell})] } ds \right|,
    \end{multline*}
    where $H(s)=G_1(s,x(s),\L(s,\xx^0_s))$. Due to the assumptions on $G_1$ and $K$, it follows that the components of $H$ are essentially bounded and therefore belong to $L^1([t_0,t_0+\md],\R)$. Let $h_{i}(\cdot)$ be the $i$-th row of $H(\cdot)$ and write $h_i(\cdot)=[h_{i1}(\cdot)\cdots h_{i\ell}(\cdot)]$, where each $h_{im}(\cdot)$ has $n$ components, for $m=1,\ldots,\ell$. Then, if $\rho_i(s)$ is the $i$-th component of the product $H(s)[\xx^k(s-\tau_{1:\ell})-\xx^0(s-\tau_{1:\ell})]$, 
    \begin{align*}
    \int_{t_0}^t \rho_i(s) ds = \sum_{m=1}^{\ell} \int_{t_0}^t h_{im}(s) [\xx^k(s-\tau_m)-\xx^0(s-\tau_m)] ds.
    \end{align*}
The existence of $j$ so that 
\begin{align*}
 \left|\int_{t_0}^{t} \big[{\scriptstyle f(s,\xx^k(s-\tau_{0:\ell}),\xx^0_s) - f(s,\xx^0(s-\tau_{0:\ell}),\xx^0_s)}\big]\ dt \right|<\varepsilon
 \end{align*}
 for all $t\in [t_0,t_0+\md]$ and all $k\ge j$ will follow once we show that each sequence of functions $\{I_{im}^k(\cdot)\}$ with
 $I_{im}^k(t) := \int_{t_0}^t h_{im}(s)[\xx^k(s-\tau_m)-\xx^0(s-\tau_m)] ds$ converges to 0 uniformly in $t\in [t_0,t_0+\md]$. Let $t\in [t_0,t_0+\md]$. Considering  the change of variable $s=r+t_0+\md$, the function $h_{im}^*(r)=h_{im}(r+t_0+\md)$ and $\Delta_m=\md-\tau_m$, we have that
 \begin{multline*}
    \int_{t_0}^t h_{im}(s)[\xx^k(s-\tau_m)-\xx^0(s-\tau_m)] ds\\=\int_{-\md}^{t-t_0-\md} h^*_{im}(r)[\xx^k_{t_0+\Delta_m}(r)-\xx^0_{t_0+\Delta_m}(r)]dr\\
    =\int_{-\md}^0 h^*_{im}(r)\chi(t,r)[\xx^k_{t_0+\Delta_m}(r)-\xx^0_{t_0+\Delta_m}(r)]dr,
 \end{multline*}
 where $\chi(t,r)=1$ if $-\md \le r \le t-t_0-\md$ and $0$ elsewhere. Then, the uniform convergence of each integral follows from the facts that the maps $\L_{im}^*(t,\phi)=\int_{-\md}^0 h_{im}^*(r)\chi(t,r)\phi(r)dr$ are continuous on $[t_0,t_0+\md]\times B_R^{L^{\infty}}$ when the weak-* topology is considered in $B_R^{L^{\infty}}$, and therefore are uniformly continuous there, and the uniform convergence of $\xx^k_{t_0+\Delta_m}$ to $\xx^0_{t_0+\Delta_m}$ given by Lemma \ref{lem:unifdist}. Assumption~\ref{as:fweaklim} is hence satisfied.
\end{proof}

The next result establishes that solutions corresponding to weakly converging sequences of initial conditions converge uniformly to the solution corresponding to the limit of the initial conditions' sequence over compact time intervals, and that this happens uniformly for all systems within a family provided the reachability sets of every system in the family are uniformly bounded. The proof is given in Appendix~\ref{ap:proof}.
\begin{prop}
  \label{prop:unifxconv-rel}
  Let a family $\F$ satisfy Assumption~\ref{as:fweaklim}. Let $t_0 \ge 0$, $\xi^0 \in \R^n$, $\phi^0 \in L^\infty$. For every $f\in\F$, let the unique maximally defined solution $x(\,\cdot\,; t_0, \xi^0,\phi^0,f)$ of~(\ref{eq:system-pt-dist-ess}) have existence time $T_{(t_0,\xi^0,\phi^0,f)} > t_0$. Let $T \in (t_0,T_{(t_0, \xi^0,\phi^0,f)})$ for every $f\in\F$, let $\xi^k \to \xi^0$ and $\phi^k \wto \phi^0$. If
  \begin{align}
    \label{eq:Rx-def}
    R_x &:= \sup_{f\in\F} \left(\sup_{t_0\le t\le T} |x(t;t_0,\xi^0,\phi^0,f)|\right) < \infty    
  \end{align}
  then for every $\varepsilon > 0$ there exists $j\in\N$ such that for all $k\ge j$ and all $f\in\F$,
  \begin{align}
    \label{eq:xunifx0-gen}
    \sup_{t_0\le s\le T} |x(s;t_0,\xi^k,\phi^k,f) - x(s;t_0,\xi^0,\phi^0,f)| \le \varepsilon 
  \end{align}
\end{prop}
\begin{rem}
  The inner supremum in~(\ref{eq:Rx-def}) is finite for every $f$, due to the continuity of the solution over the compact interval $[t_0,T]$; hence~(\ref{eq:Rx-def}) is true if the cardinality of $\F$ is finite. However, this result is most useful when $\F$ contains an infinite number of systems and may be used in the future for systems with inputs. 
\end{rem}

\subsection{Bounded Reachability Sets}
\label{sec:BRS}

The next theorem and corollary constitute our main results. The uniform boundedness of all the solutions of~\eqref{eq:system-pt-dist-ess} from a specific initial time and bounded initial conditions is the result of Theorem~\ref{thm:brs} below. The subsequent Corollary~\ref{cor:FCiffBRS} establishes BRS. 
\begin{thm}
  \label{thm:brs}
  Consider a single-element family $\F=\{f\}$ satisfying Assumption~\ref{as:fweaklim}. For each $t_0 \ge 0$, $\xi\in\R^n$ and $\phi\in L^\infty$,
  let $x(\,\cdot\,; t_0, \xi, \phi) : [t_0,T_{(t_0,\xi,\phi)}) \to \R^n$ denote the unique maximally defined solution of~(\ref{eq:system-pt-dist-ess}). Let $T>t_0\ge 0$, $R>0$, and suppose that 
    \begin{align}
        \label{eq:thm:brs:exist-time}
      \bar T := \inf_{\xi\in B_R^n, \phi\in B_R^{L^\infty}} T_{(t_0,\xi,\phi)} > T.
    \end{align}
    Then, \eqref{eq:brs} holds.
\end{thm}
\begin{proof}
  For a contradiction, suppose that the supremum in~(\ref{eq:brs}) equals $\infty$. 
  Then, sequences $\{t^k\}, \{\xi^k\}, \{\phi^k\}$ must exist, with $t^k \in [t_0,T]$, $\xi^k\in B_R^n$, and $\phi^k\in B_R^{L^\infty}$, such that
  \begin{align*}
    \lim_{k\to\infty} |x(t^k;t_0,\xi^k,\phi^k)| = \infty.
  \end{align*}
  From compactness of the interval $[t_0,T]$ and the closed ball $B_R^n$, and from Lemma~\ref{lem:weakfacts}, subsequences $\{t^{k_i}\}$, $\{\xi^{k_i}\}$ and $\{\phi^{k_i}\}$ must exist so that $t^{k_i} \to t^0 \in [t_0,T]$, $\xi^{k_i} \to \xi^0 \in B_R^n$, $\phi^{k_i} \wto \phi^0 \in B_R^{L^\infty}$ and
  \begin{align}
    \label{eq:limunb}
    \lim_{i\to\infty} |x(t^{k_i}; t_0,\xi^{k_i},\phi^{k_i})| = \infty.
  \end{align}
  Let $T' \in (T,\bar T)$ and define
  \begin{align*}
    R_x = \sup_{t_0\le s \le T'} |x(s; t_0, \xi^0,\phi^0)|
  \end{align*}
  which is finite by continuity and the fact that $T_{(t_0,\xi^0,\phi^0)} \ge \bar T > T'$. 
  According to Proposition~\ref{prop:unifxconv-rel}, there exists $j\in\N$ such that
  \begin{align*}
    \sup_{t_0\le s\le T'} |x(s;t_0,\xi^{k_i},\phi^{k_i}) - x(s;t_0,\xi^0,\phi^0)| \le 1
  \end{align*}
  for all $i\ge j$ and hence
  \begin{align*}
    \sup_{t_0\le s\le T'} |x(s;t_0,\xi^{k_i},\phi^{k_i})| \le R_x + 1 \qquad \tforall i\ge j
  \end{align*}
  This contradicts~(\ref{eq:limunb}) because $t^{k_i} \in [t_0,T] \subset [t_0,T']$ for all $i\in\N$.
\end{proof}

\begin{cor} \label{cor:FCiffBRS}
  Let $f : \R_{\ge 0} \times \R^{(\ell+1)n} \times L^\infty \to \R^n$ and let the single-element family $\F= \{f\}$ satisfy Assumption~\ref{as:fweaklim}. System~(\ref{eq:system-pt-dist-ess}) is forward complete (for every $t_0 \ge 0$) if and only if it has bounded reachability sets.
\end{cor}
\begin{proof}
    Let $T> t_0\ge 0$ and $R>0$. By the forward completeness assumption, then~\eqref{eq:thm:brs:exist-time} holds. Application of Theorem~\ref{thm:brs} then establishes~\eqref{eq:brs}. Since $t_0$, $T$ and $R$ are arbitrary, then BRS is established.
\end{proof}

System (\ref{eq:system-pt-dist-ess}) is said to have BRS for continuous initial conditions if (\ref{eq:brs}) holds with initial condition restricted as $(\xi,\phi) \in \mathcal{C}$ (see the discussion at the end of Section \ref{sec:equiv-init-cond}).
\begin{prop} \label{prop:ciffm}
  Let $f : \R_{\ge 0} \times \R^{(\ell+1)n} \times L^\infty \to \R^n$ and let the single-element family $\F= \{f\}$ satisfy Assumption~\ref{as:fweaklim}. Then, system~(\ref{eq:system-pt-dist-ess}) has BRS if and only if it does for continuous initial conditions.
  \end{prop}
  \begin{proof} The \emph{only if} part is straightforward, so the \emph{if} part is next established. Let $0\le t_0 < T$ and $R>0$. Since  system~(\ref{eq:system-pt-dist-ess}) has BRS for continuous initial conditions, the supremum in (\ref{eq:brs}) is equal to a nonnegative number, say $M$, when the initial conditions are restricted to be continuous. Let $(\xi_0,\phi)$ be an initial condition such that $|\xi_0|\le R$ and $\|\phi\|\le R$. Let $x:[t_0,T_x)\to \R^n$ be the maximal solution of (\ref{eq:system-pt-dist-ess}) corresponding to this initial condition. By using standard arguments of the theory of real functions, it can be proved that there exists a sequence of continuous functions $\{\psi^k\}$ such that $\|\psi^k\|\le R$, $\psi^k(0)=\xi_0$ and $\psi^k\to \phi$ a.e. on $[-\md,0]$. It then follows that $\psi^k\wto \phi$. If $x^k(\cdot)$ is the maximal solution of (\ref{eq:system-pt-dist-ess}) corresponding to the initial condition $(\xi_0,\psi^k)$, we have that $x^k$ is defined for all $t\ge t_0$, $|x^k(t)|\le M$ for all $t\in [t_0,T]$ and all $k$. In addition, applying Proposition~\ref{prop:unifxconv-rel}, it follows that $x^k\to x$ uniformly on $[t_0,T^\prime]$ for all $T^{\prime}\in [t_0,T]\cap [t_0,T_x)$. In consequence $|x(t)|\le M$ for all $t\in [t_0,T^{\prime}]$ and for all $T^{\prime}\in [t_0,T]\cap [t_0,T_x)$. The latter implies that $T_x>T$, and since $T$ is arbitrary, that $T_x=\infty$. Therefore, system (\ref{eq:system-pt-dist-ess}) is forward complete and has BRS.
  \end{proof}
The following result is a straightforward consequence of Corollary~\ref{cor:FCiffBRS} and Proposition~\ref{prop:ciffm}.
\begin{cor} \label{cor:fciffbrsc}
 Let $f : \R_{\ge 0} \times \R^{(\ell+1)n} \times L^\infty \to \R^n$ and let the single-element family $\F= \{f\}$ satisfy Assumption~\ref{as:fweaklim}. Then, system~(\ref{eq:system-pt-dist-ess}) has bounded reachability sets for continuous initial conditions if and only if it is forward complete for initial conditions in $\R^n\times L^{\infty}$.
 \end{cor}
 \begin{rem} In many applications, the goal is to establish that a system possesses the BRS property for continuous initial conditions. Corollary~\ref{cor:fciffbrsc} shows that, for systems satisfying Assumption~\ref{as:fweaklim}, it is sufficient to prove the solutions are defined for all future times when the initial conditions lie in $\L^{\infty}$, which is, a priori, a simpler task. Reciprocally, if for some initial condition in $\L^{\infty}$ and some initial time the solution has a finite existence time, then the system cannot have BRS for continuous initial conditions. 
 \end{rem}

\subsection{Discussion}
\label{sec:discussion}

In the absence of discrete delays, i.e. when 
\begin{align*}
  \dot x(t) = f(t,x(t),(\phi\conc{t_0}x)_t)
\end{align*}
with $\phi \mapsto f(t,\xi,\phi)$ being insensitive to differences over zero-measure sets, then
item~\ref{item:discint-cont}) of Assumption~\ref{as:fweaklim} becomes trivially satisfied and only~\ref{item:weakstar-cont}) imposes a constraint on $f$. In contrast, in the absence of distributed delays, i.e. when 
\begin{align*}
  \dot x(t) = f(t,x(t),x(t-\tau_1),\ldots,x(t-\tau_\ell)),
\end{align*}
then only~\ref{item:discint-cont}) of Assumption~\ref{as:fweaklim} imposes a constraint. For systems that are in addition time invariant, i.e. for 
\begin{align*}
  \dot x(t) = f(x(t),x(t-\tau_1),\ldots,x(t-\tau_\ell)),
\end{align*}
then it has already been established \cite{bricha_lcss24} that if the system is forward complete for initial conditions in $\R^n \times L^\infty$ (or equivalently in $\L^{\infty})$, then it has bounded reachability sets. However, the proof strategy of \cite{bricha_lcss24}, proposed originally in \cite{manhai_auto24}, does not apply when the system is time-varying or it is time-invariant  but distributed delays are present. 

The property that allows to establish BRS from knowledge of forward completeness is that the solution maps $S : \R^n \times L^\infty \to \mathcal{C}([t_0,T],\R^n)$
$$ (\xi,\phi) \mapsto S(\xi,\phi) = x(\,\cdot\,;t_0,\xi,\phi) $$
are continuous for every $0 \le t_0 < T$ when the topology of uniform convergence is considered in $\mathcal{C}([t_0,T],\R^n)$ and the weak-* topology is considered in $L^\infty$. Since every closed and bounded ball in $\R^n \times L^\infty$ is then compact (Lemma~\ref{lem:weakfacts}), the continuity of the solution map is sufficient to ensure the uniform boundedness of the corresponding set of solutions, and hence BRS. This type of continuity is precisely the result of Proposition~\ref{prop:unifxconv-rel}, while the uniform boundedness of the set of solutions is given by Theorem~\ref{thm:brs}.

In this light, the most difficult point is to provide checkable and applicable conditions that ensure the continuity of the solution maps. These conditions are those of Assumption~\ref{as:fweaklim}, where the type of systems that satisfy it were provided by Proposition~\ref{prop:exa:ass3}.

\section{Conclusions}
\label{sec:conclusions}

Sufficient conditions for the boundedness of the reachability sets of nonlinear systems involving both pointwise and distributed delays were developed, considering initial conditions that are not restricted to being continuous. The conditions provided involve the concepts of weak-* convergence in the Banach space of equivalence classes of essentially bounded functions under the essential supremum norm. Future work is to provide results for time-invariant systems with inputs by considering the possible behaviors as those of a family of time-varying systems without inputs. 
Interesting questions for future research are to what extent the conditions can be relaxed and whether conditions based on different concepts can also be developed.

\appendices
\section{Some Technical Proofs} 
\subsection{Proof of Proposition~\ref{prop:exa:ass2}}
\label{app:ex:ass2}
Let $T>0$ and $R>0$. From items \ref{item:caratheodory-G}) and \ref{item:Lips-G}), $\mathcal{N}$ is well-defined, and for all $\phi$, $\psi\in B_R^{L^{\infty}}$ and $t\in [0,T]$ we have that
\begin{align*}
|\Nin{t,\phi}|\le \|M_{R,T}\|_1
\end{align*}
and
\begin{align*}
|\Nin{t,\phi}-\Nin{t,\psi}|\le \|\hat{L}_{R,T}\|_1\|\phi-\psi\|.
\end{align*}
Therefore, item \ref{item:fLip-pt-dist}) of Assumption \ref{as:Lip-pt} follows from the latter and \ref{ex:lipschitz}). Since for all $t\in [0,T]$
\begin{align*}
|f(t,0, \zero)|&=|g(t,0,\Nin{t,\zero})|\\
&\le |g(t,0,\Nin{t,\zero})-g(t,0,0)|+|g(t,0,0)|,
\end{align*}
items \ref{ex:caratheodory}) and \ref{ex:lipschitz}), and the fact that $\Nin{t,\zero}$ is uniformly bounded on $[0,T]$ imply that $f(\cdot,0,\zero)$ is essentially bounded on $[0,T]$. So, item \ref{item:f0bound}) of Assumption \ref{as:Lip-pt} follows. Let $t_0\ge 0$, $T>t_0$ and $x:[t_0,T]\to \R^n$ be continuous. Suppose first that $\phi \in \mathcal{C}$ is such that $\phi(0)=x(t_0)$. Then $\xx = (\phi\conc{t_0}x)$ is continuous on the compact interval $[t_0-\md,T]$, and hence uniformly continuous. Let  $\{t^k\}\subset [t_0,T)$ be such that $t^k \to t\in [t_0,T]$. 
Pick any $R>0$ such that $\|\xx_s\|\le R$ for all $s\in [t_0,T]$. Then, since
\begin{align*}
|\Nin{t^k,\xx_{t^k}}-\Nin{t^k,\xx_{t}}|\le \|\hat{L}_{R,T}\|_1\|\xx_{t^k}-\xx_{t}\|
\end{align*}
and $\lim_{k\to \infty}\|\xx_{t^k}-\xx_t\|=0$ because of the uniform continuity of $\xx$, it follows that
\begin{align}
    \label{eq:ex1:Ntkxxk0}
    \lim_{k\to\infty} |\Nin{t^k,\xx_{t^k}}-\Nin{t^k,\xx_{t}}|= 0. 
\end{align}
Since $G$ is continuous in its first variable, then for all $s\in [-\md,0]$,
\begin{align}
    \label{eq:ex1:dG0}
    \lim_{k\to\infty} |G(t^k,s,\xx_t(s))-G(t,s,\xx_t(s))| = 0,
\end{align}
and for almost all $s\in [-\md,0]$ and all $k$, 
\begin{align}
    |G(t^k,s,\xx_t(s))-G(t,s,\xx_t(s))|\le 2M_{R,T}(s) 
\end{align}
follows from~\ref{item:Lips-G}). Since
\begin{multline*}
|\Nin{t^k,\xx_t}-\Nin{t,\xx_t}|\\ \le
\int_{-\md}^0 |G(t^k,s,\xx_t(s))-G(t,s,\xx_t(s))| ds,
\end{multline*}
by the Lebesgue dominated convergence theorem, then 
\begin{align}
    \label{eq:ex1:Ntk0}
    \lim_{k\to\infty} |\Nin{t^k,\xx_t}-\Nin{t,\xx_t}| = 0.
\end{align} 
From \eqref{eq:ex1:Ntkxxk0}, \eqref{eq:ex1:Ntk0}, and the triangle inequality, it follows that
\begin{align*}
    \lim_{k\to\infty} |\Nin{t^k,\xx_{t^k}}-\Nin{t,\xx_{t}}| = 0 
\end{align*}
and hence the map $t \mapsto \Nin{t,\xx_t}$ is continuous. Since $g$ is Carath\'eodory and $\xx(\,\cdot-\tau_i)$ is Lebesgue measurable for $i=0,\ldots,\ell$, it follows that the map 
\begin{align*}
    t\mapsto g(t,\xx(t-\tau_{0:\ell}),\Nin{t,\xx_t}) = f(t,\xx(t-\tau_{0:\ell}), \xx_t)
\end{align*}
is Lebesgue measurable. Suppose next that $\phi \in L^ {\infty}$. Then, there exists a bounded sequence of continuous functions $\phi^k$ such that $\phi^k\to \phi$ almost everywhere and $\phi^k(0)=x(t_0)$. Define $\xx^k=\phi^k \conc{t_0}x$ and $\xx=\phi \conc{t_0} x$. Then, $\xx^k(\,\cdot-\tau_i)\to \xx(\,\cdot-\tau_i)$ almost everywhere for $i=0,\ldots,\ell$, and $\xx^k_t(s)\to \xx_t(s)$ for almost all $s\in [-\md,0]$ for all $t\in [t_0,T]$. The fact that $G(t,s,\cdot)$ is continuous then implies that $G(t,s,\xx^k_t(s))\to G(t,s,\xx_t(s))$ for almost all $s\in [-\md,0]$. Since $\{\xx^k_t\}$ is bounded in $L^{\infty}$, say $\|\xx^k_t\|\le R$ for all $k$, $|G(t,s,\xx^k_t(s))|\le M_{R,T}(s)$ for almost all $s$.  Applying the Lebesgue dominated convergence theorem we then have that
$\Nin{t,\xx^k_t}\to \Nin{t,\xx_t}$.
Since $f(t,\xx^k(t-\tau_{0:\ell}), \xx^k_t)$ is Lebesgue measurable for all $k$, and $g(t,\xx^k(t-\tau_{0:\ell}),\Nin{t,\xx^k_t})\to g(t,\xx(t-\tau_{0:\ell}),\Nin{t,\xx_t})$ for almost all $t\in [t_0,T]$ because of the continuity of $g$ with respect to the last two variables, it follows that the map $t\mapsto f(t,\xx(t-\tau_{0:\ell}),\xx_t)$ is Lebesgue measurable.

\subsection{Proof of Proposition~\ref{prop:unifxconv-rel}}
\label{ap:proof}
We first prove the following result.
\begin{lem}
  \label{lem:fweakint}
  Let the family $\F$ satisfy Assumption~\ref{as:fweaklim}. Then, the following holds: for every $\xi^0 \in \R^n$, sequence $\{\phi^k\}$ satisfying $\phi^k \wto \phi^0 \in L^\infty$, $t_0 \ge 0$, $T\in [t_0,t_0+\md]$, and $\varepsilon > 0$, there exists $j\in\N$ such that for all $t \in [t_0,T]$, all $k\ge j$ and all $f\in\F$, 
    \begin{subequations}
      \label{eq:wLipint}
      \begin{align}
        \label{eq:wLipinta}
        &\left|\int_{t_0}^t f^k(q) - f^0(q)\ dq \right| \le \varepsilon,\\
        \label{eq:fk-xx}
        &f^k(q) = f(q,\xx^k(q-\tau_{0:\ell}), \xx^k_q),\quad
        \xx^k = (\phi^k\conc{t_0} x)\\
        &\text{where }x(\cdot) = x(\,\cdot\,; t_0, \xi^0, \phi^0, f)
      \end{align}
    \end{subequations}
    denotes the unique solution of (\ref{eq:system-pt-dist-ess}) with $\xi_0 = \xi^0$.
\end{lem} 
\begin{proof} Let a sequence $\{\phi^k\}$ satisfy $\phi^k \wto \phi^0 \in L^\infty$, let $t_0 \ge 0$, $T \in [t_0,t_0+\md]$, and consider an absolutely continuous function $x : [t_0,T] \to \R^n$. Define 
  \begin{align*}
    R := 1 + \max\left\{ \sup_{t\in [t_0,T]}|x(t)|, \sup_{k\in\N} \|\phi^k\|  \right\}
  \end{align*}
  which is finite by continuity of $x$ and by Lemma~\ref{lem:weakfacts}. Then, $x:[t_0,T] \to B_R^n$, $\phi^k \in B_R^{L^\infty}$ for all $k$ and $\phi^0 \in B_R^{L^\infty}$ by Lemma~\ref{lem:weakfacts}.
  Let $Z$ be the zero-measure set given by Assumption~\ref{as:fweaklim}\ref{item:weakstar-cont}) in correspondence with $T$ and $R$. Assumption~\ref{as:fweaklim}\ref{item:weakstar-cont}) and the fact that the weak-* topology in $B_R^{L^\infty}$ is induced by the metric~(\ref{eq:normstar})--(\ref{eq:def:metric}) imply the following. 
  \begin{claim}
    \label{clm:Bmetrizable}
    For every $\varepsilon_1 > 0$ there exists $\delta_1 > 0$ for which
    \begin{align*}
      |f(t,\xi_{0:\ell},\phi) - f(t,\xi_{0:\ell},\psi)| < \varepsilon_1
    \end{align*}
    for all $t\in [0,T]\setminus Z_f$, $f\in\F$, $\xi_{0:\ell} \in B_R^{(\ell+1) n}$, $\phi,\psi\in B_R^{L^\infty}$ with $\norm{\phi-\psi} < \delta_1$.
  \end{claim}
  \textit{Proof of Claim~\ref{clm:Bmetrizable}:}
  For a contradiction, suppose that there exists $\varepsilon_1>0$ and a decreasing sequence of positive values $\delta^k \to 0$ such that in correspondence with every $\delta^k > 0$ there exist $f^k\in\F$, $t^k\in [0,T]\setminus Z_{f^k}$, $\xi_{0:\ell}^k \in B_R^{(\ell+1) n}$ and $\theta^k,\psi^k\in B_R^{L^\infty}$ with $\norm{\theta^k-\psi^k} < \delta^k$ for which
  \begin{align*}
    |f^k(t^k,\xi_{0:\ell}^k,\theta^k) - f^k(t^k,\xi_{0:\ell}^k,\psi^k)| \ge \varepsilon_1.
  \end{align*}
  From boundedness, it follows that subsequences $\xi_{0:\ell}^{k_i}$, $\theta^{k_i}$ and $\psi^{k_i}$ exist such that $\xi_{0:\ell}^{k_i} \to \xi_{0:\ell}^* \in B_R^{(\ell+1) n}$, $\theta^{k_i} \wto \theta^* \in B_R^{L^\infty}$ and $\psi^{k_i} \wto \psi^* \in B_R^{L^\infty}$. Since $\norm{\theta^{k_i}-\psi^{k_i}} < \delta^{k_i} \to 0$, then $\theta^* = \psi^*$. Then, from Assumption~\ref{as:fweaklim}\ref{item:weakstar-cont}) there exist $j_1,j_2\in\N$ such that
  \begin{align*}
    |f^{k_i}(t^{k_i},\xi_{0:\ell}^{k_i},\theta^{k_i}) - f^{k_i}(t^{k_i},\xi_{0:\ell}^{k_i},\theta^*)| < \varepsilon_1/2
  \end{align*}
  for all $i\ge j_1$ and, for all $i\ge j_2$, also
  \begin{align*}
    |f^{k_i}(t^{k_i},\xi_{0:\ell}^{k_i},\psi^{k_i}) - f^{k_i}(t^{k_i},\xi_{0:\ell}^{k_i},\theta^*)| < \varepsilon_1/2.
  \end{align*}
  As a consequence, 
  \begin{align*}
    |f^{k_i}(t^{k_i},\xi_{0:\ell}^{k_i},\theta^{k_i}) - f^{k_i}(t^{k_i},\xi_{0:\ell}^{k_i},\psi^{k_i})| < \varepsilon_1
  \end{align*}
  for all $i\ge j=\max\{j_1,j_2\}$, reaching a contradiction and estalishing Claim~\ref{clm:Bmetrizable}.

  Let $\varepsilon > 0$. From Claim~\ref{clm:Bmetrizable}, in correspondence with $\varepsilon_1 = \varepsilon/(2\md)$, there exists $\delta_1 > 0$ such that $|f(t,\xi_{0:\ell},\phi) - f(t,\xi_{0:\ell},\psi)| < \varepsilon/(2\md)$ for all $f\in\F$, $t\in [0,T]\setminus Z_f$,  $\xi_{0:\ell} \in B_R^{(\ell+1) n}$ and $\norm{\phi-\psi} < \delta_1$. From Lemma~\ref{lem:unifdist}, for $\varepsilon_2 = \delta_1$ there exists $q_2 \in \N$ such that $\norm{\xx^k_t - \xx^0_t} < \varepsilon_2$ whenever $t\in [t_0,T]$ and $k\ge q_2$. As a consequence, for all $k\ge q_2$, it is true that 
  \begin{align}
    \label{eq:finc-p4}
    \big|f(t,\xi_{0:\ell},\xx^k_t) - f(t,\xi_{0:\ell},\xx^0_t)\big| < \varepsilon/(2\md)
  \end{align}
  holds uniformly for all $f\in\F$, $t\in [t_0,T]\setminus Z_f$ and $\xi_{0:\ell}\in B_R^{(\ell+1) n}$.
  Note that $|\xx^k(t-\tau_i)| \le R$ holds for almost all $t\in [t_0,T]$ for $i=1,\ldots,n$ and for all $t\in [t_0,T]$ for $i=0$. The latter fact and~(\ref{eq:finc-p4}) imply that for every $f\in\F$ and $k\in\N$, there exists a zero-measure set $Z_{f,k} \subset [t_0,T]$ such that
  \begin{align}
    \label{eq:finc-p5}
    \big|f(t,\xx^k(t-\tau_{0:\ell}),\xx^k_t) - f(t,\xx^k(t-\tau_{0:\ell}),\xx^0_t)\big| < \frac{\varepsilon}{2\md}
  \end{align}
  is true for all $t\in [t_0,T]\setminus Z_{f,k}$. In addition, from Assumption~\ref{as:fweaklim}\ref{item:discint-cont}), there exists $q_3 \in \N$ such that
  \begin{align}
    \label{eq:finc-discrete}
    \left| \int_{t_0}^{t} [{\scriptstyle f(t,\xx^k(t-\tau_{0:\ell}),\xx^0_t) - f(t,\xx^0(t-\tau_{0:\ell}),\xx^0_t)}]\ dt \right| < \frac{\varepsilon}{2}
  \end{align}
  for all $t\in [t_0,T]$, $k\ge q_3$ and $f\in\F$.
  Then, for all $t\in [t_0,T]$ and $k\ge w := \max\{q_2,q_3\}$,
  \begin{multline*}
    \left| \int_{t_0}^{t} [{\scriptstyle f(t,\xx^k(t-\tau_{0:\ell}),\xx^k_t) - f(t,\xx^0(t-\tau_{0:\ell}),\xx^0_t)}]\ dt \right| \le \\
    \int_{t_0}^{t} |{\scriptstyle f(t,\xx^k(t-\tau_{0:\ell}),\xx^k_t) - f(t,\xx^k(t-\tau_{0:\ell}),\xx^0_t)}|\ dt + \\
    \left| \int_{t_0}^{t} [{\scriptstyle f(t,\xx^k(t-\tau_{0:\ell}),\xx^0_t) - f(t,\xx^0(t-\tau_{0:\ell}),\xx^0_t)}]\ dt\right|  < \varepsilon
  \end{multline*}
  where~(\ref{eq:finc-p5}), the fact that $t - t_0 \le \md$ and~(\ref{eq:finc-discrete}) were used. 
  \end{proof}
  Next, Proposition \ref{prop:unifxconv-rel} is proved.
  Define
  \begin{subequations}
      \label{eq:R0-def}
      \begin{align}
        R_\phi &:= \sup_{k\in\N} \max\{|\xi^k|,\|\phi^k\|, \|\phi^0\|\},\\
        R &:= \max\{R_\phi,R_x,1\},\quad
        \label{eq:Fdef}
        F := 2 R + 1. 
      \end{align}
  \end{subequations}
  Note that $R_\phi$ is finite due to Lemma~\ref{lem:weakfacts}. 
  Let $\alpha : \R^n \to [0,1]$ be a smooth function satisfying
  \begin{align}
    \label{eq:h:alpha:def}
    \begin{cases}
      \alpha(\zeta) = 1 &\text{if }|\zeta|\le 2R,\\
      \alpha(\zeta) = 0 &\text{if }|\zeta|\ge F,
    \end{cases}
  \end{align}
  and define the functions $h_f : \R_{\ge 0} \times \R^{(\ell+1)n} \times L^\infty \to \R^n$ as
  \begin{align*}
    h_f(t,\zeta_{0:\ell},\psi) = \alpha(\zeta_0) f(t,\zeta_{0:\ell},\psi)
  \end{align*}
  Since $\alpha$ is smooth, and $\alpha(\zeta)$ is zero for all $|\zeta|\ge F$, then there exists $D > 0$ such that
  \begin{align*}
    |\alpha(\zeta^1)-\alpha(\zeta^2)| \le D |\zeta^1 - \zeta^2|
  \end{align*}
  holds for all $\zeta^1,\zeta^2 \in \R^n$. As a consequence, it is straightforward to check that the family $\H = \{h_f : f\in\F\}$ also satisfies Assumption~\ref{as:Lip-pt} uniformly, and that for every $f\in\F$, $\zeta_0 \in \R^n$ and $\psi\in L^\infty$, the system 
  \begin{align*}
    \dot z(t) = h_f(t,(\psi\conc{t_0}z)(t-\tau_{0:\ell}),(\psi\conc{t_0}z)_t), \quad
    z(t_0) = \zeta_0 
  \end{align*}
  has a unique maximally defined solution $z(t)=z(t;t_0,\zeta_0,\psi,f)$. Moreover, since $\dot z(t)$ is nonzero only if $z(t)$ is in a compact set, then $z(t)$ exists for all $t\ge t_0$. In addition, if $(\zeta_0,\psi) \in B_{2R}^n \times B_{2R}^{L^\infty}$, then $z(t) \in B_F^n$ for all $t\ge t_0$, and as long as $z(t) \in B_{2R}^n$, then $z(t) = x(t; t_0, \zeta_0,\psi,f)$. In particular, $z^0(t) := z(t; t_0,\xi^0,\phi^0,f) = x(t; t_0, \xi^0,\phi^0,f) =: x^0(t) $ for all $t\in [t_0,T]$.

  Let $L_1=L_1(F,2T)$ be the Lipschitz constant of the family $\H$ according to Assumption~\ref{as:Lip-pt}.
  Let $\varepsilon > 0$, $\varepsilon' = \min\{\varepsilon,1/2\}$ and $\varepsilon_1 = \varepsilon' e^{-L_1 T}/2$. Let $\ell_1 \in \N$ be such that
  \begin{align}
    |\xi^k - \xi^0| < \varepsilon_1\qquad\tforall k\ge \ell_1
  \end{align}
  Employ the notation $z^i(t) := z(t;t_0,\xi^i,\phi^i,f)$ for every $i\in \N_0$ and $f\in\F$, $\Delta z^i(t) := z^i(t) - z^0(t)$, $\zz^{i,k} := (\phi^i \conc{t_0} z^k)$ (any representative of $\phi^i$ is equivalent), and let the function $g$ be defined via $g(q,\zz) := h_f(q,\zz(t-\tau_{0:\ell}),\zz_t)$. With this shorthand notation, then $h_f(q,(\phi^k\conc{t_0}x^0)(q-\tau_{0:\ell}),(\phi^k\conc{t_0}x^0)_q)$ can be abbreviated as $g(q,\zz^{k,0})$.
  Let $\ell_2\in\N$ be given by Lemma~\ref{lem:fweakint} in correspondence with $\xi^0$, the sequence $\{\phi^k\}$, $t_0$, $\min\{t_0+\md,T\}$ and $\varepsilon_1$. Note that $\ell_2$ depends on the family $\F$ and not on a specific $f\in\F$. The difference between solutions $z^k$ and $z^0$ can be bounded as follows:
  \begin{align}
    &|\Delta z^k(t)| 
     \le |\xi^k-\xi^0| + \left| \int_{t_0}^t [g\big(q, \zz^{k,k} \big) -
    g\big(q, \zz^{0,0} \big)] \ dq \right| \notag \\
    & \le |\xi^k-\xi^0| + \int_{t_0}^t \Big| g\big(q, \zz^{k,k} \big) -
    g\big(q, \zz^{k,0} \big) \Big| dq\notag \\
    & \quad 
    + \left| \int_{t_0}^t [g\big(q, \zz^{k,0} \big) -
      g\big(q, \zz^{0,0} \big)] \ dq \right|
    \label{eq:Deltbnd1}
  \end{align}
  Note that $g(q,\zz^{k,0}) = f(q,(\phi^k\conc{t_0}x^0)(q-\tau_{0:\ell}),(\phi^k\conc{t_0}x^0)_q)$ for all $k \in \N_0$ and $q\in [t_0,T]$ because $|(\phi^k\conc{t_0}x^0)(q)| \le R$ for all $q\in [t_0,T]$ (and all $f\in\F$) causes $\alpha((\phi^k\conc{t_0}x^0)(q)) = 1$, according to~(\ref{eq:h:alpha:def}). From~(\ref{eq:Deltbnd1}), then
  \begin{align}
    \label{eq:Lipbnd2}
    |\Delta z^k(t)|
    &\le \varepsilon_1 + \int_{t_0}^t L_1 \sup_{t_0 \le s \le q} |\Delta z^k(s)| dq + \varepsilon_1, 
  \end{align}
  where Assumption~\ref{as:Lip-pt}\ref{item:fLip-pt-dist}) for $h_f$ has been used, which holds uniformly for all $f\in\F$, for all $k\ge j := \max\{\ell_1,\ell_2\}$ and all $t\in [t_0,T]$, the latter because if $T>t_0+\md$, then $g(q,\zz^{k,0}) - g(q,\zz^{0,0}) = 0$ for all $q \in [t_0+\md,T]$. 
  Taking the supremum, from~(\ref{eq:Lipbnd2}) one can reach
  \begin{align*}
    \sup_{t_0\le s \le t} |\Delta z^k(s)| \le
    2\varepsilon_1 + L_1 \int_{t_0}^t \sup_{t_0\le s\le q} |\Delta z^k(s)| dq
  \end{align*}
  and through application of Gronwall inequality,
  \begin{align}
    \sup_{t_0\le s \le t} |\Delta z^k(s)| &\le 2 \varepsilon_1 e^{L_1 (t-t_0)} 
    \label{eq:delta-result}
    \le 2\varepsilon_1 e^{L_1 T} = \varepsilon' \le \varepsilon
  \end{align}
  which holds for all $t\in [t_0,T]$, all $k\ge j$ and all $f\in\F$. 
  Then, it should be true that for all $k\ge j$,
  \begin{align*}
    |z^k(t)| &\le
    \sup_{t_0\le t \le T} |z^0(t) + \Delta z^k(t) |
    \le R_x + \varepsilon' < 2R
  \end{align*}
  and therefore $z^k(t) = x(t; t_0,\xi^k,\phi^k,f)$ for all $t\in [t_0,T]$ and $f\in\F$, and the result is established from~(\ref{eq:delta-result}).



\bibliographystyle{plain}
\bibliography{FCimpliesBRS_input_TAC}

\end{document}